\pdfoutput=1
\documentclass[a4paper, 12pt]{article}

\usepackage[backend=biber]{biblatex}
\usepackage{amsmath,amssymb}
\usepackage{graphicx}
\usepackage{bm}
\usepackage{xcolor}
\usepackage[colorlinks,linkcolor=darkblue,citecolor=darkblue,linktocpage,hypertexnames=true,pdfpagelabels]{hyperref}
\usepackage[utf8]{inputenc}
\usepackage[cal=boondoxo]{mathalfa}
\usepackage[english]{babel}
\usepackage{mathtools}
\usepackage[textsize=footnotesize]{todonotes}  

\usepackage[capitalise]{cleveref}
\usepackage{nameref}

\usepackage{subcaption}
\usepackage{hyphenat}
\usepackage{verbatim}

\usepackage{float}
\usepackage{algorithm}
\usepackage[noend]{algpseudocode}
\makeatletter
\algrenewcommand\ALG@beginalgorithmic{\ttfamily}
\makeatother

\usepackage{amsmath}
\usepackage{amssymb}
\usepackage{amsthm}

\newtheorem{theorem}{Theorem}
\newtheorem{corollary}{Corollary}
\newtheorem{lemma}{Lemma}
\newtheorem{definition}{Definition}
\newtheorem{proposition}{Proposition}

\usepackage{tikz}
\usetikzlibrary{shapes.geometric, calc, math, arrows, decorations}
\usepackage{tkz-euclide}
\usetikzlibrary{backgrounds}

\renewcommand{\epsilon}{\varepsilon}

\newcommand{\N}{\mathbb{N}}
\newcommand{\Q}{\mathbb{Q}}
\newcommand{\Z}{\mathbb{Z}}
\newcommand{\R}{\mathbb{R}}

\newcommand{\missingspinsymbol}{\blacktriangle}
\newcommand{\addressdelimiter}{\text{\textvisiblespace}}

\newcommand{\lemo}{\lhd}
\newcommand{\rimo}{\rhd}
\renewcommand{\diamond}{\lozenge}

\newcommand{\sun}{\circ}
\newcommand{\newmoon}{\bullet}

\usepackage{enumitem}

\definecolor{darkblue}{rgb}{0.0, 0.0, 0.55}

\usepackage[textsize=footnotesize]{todonotes}

\begin{document}

\title{Classical spin Hamiltonians are context-sensitive languages}

\author{\normalsize{Sebastian Stengele$^{1}$, David Drexel$^{1}$ and Gemma De las Cuevas$^{1\ast}$}\\
\small{$^1$Institute for Theoretical Physics, University of Innsbruck,}\\ \small{Technikerstr.\ 21a,  A-6020 Innsbruck, Austria} \\
\small{$^\ast$ {gemma.delascuevas@uibk.ac.at}}}

\date{{\normalsize \today}}

\maketitle

\begin{abstract}
Classical spin Hamiltonians are a powerful tool to model complex systems, characterised by a local structure given by the local Hamiltonians. One of the best understood local structures is the grammar of formal languages, which are central in computer science and linguistics, and have a natural complexity measure given by the Chomsky hierarchy. If we see classical spin Hamiltonians as languages, what grammar do the local Hamiltonians correspond to? Here we cast classical spin Hamiltonians as formal languages, and classify them in the Chomsky hierarchy. We prove that the language of (effectively) zero-dimensional spin Hamiltonians is regular, one-dimensional spin Hamiltonians is deterministic context-free, and higher-dimensional and all-to-all spin Hamiltonians is context-sensitive. This provides a new complexity measure for classical spin Hamiltonians, which captures the hardness of recognising spin configurations and their energies. We compare it to the computational complexity of the ground state energy problem, and find a different easy-to-hard threshold for the Ising model. We also investigate the dependence on the language of the spin Hamiltonian. Finally, we define the language of the time evolution of a spin Hamiltonian and classify it in the Chomsky hierarchy. Our work suggests that universal spin models are weaker than universal Turing machines.  
\end{abstract}

\clearpage

\section{Introduction}

One of the most fundamental lessons of physics is that interactions are local. 
This renders certain aspects of physical reality describable, 
as a potentially unbounded set of observations can be captured with a finite description, 
which is given by this local structure.  
In formal and natural languages, this finite description is given by a \emph{grammar}. 
Grammars are the backbone of computer science and linguistics (see, e.g.\ \cite{Ko97,Fr11b,Hi13b}), 
and the celebrated classification by Chomsky \cite{Ch65} provides a measure of their complexity.
Local interactions can thus be seen as the grammar of physical systems. 
For example, Feynman diagrams can be seen as the grammar of scattering amplitudes \cite{Ma15d},
and tensors and their `gluing' rules can be seen as the grammar of tensor networks, and by extension, of quantum many-body systems \cite{Ci20,Or18}. 

As a matter of fact, the only infinite families of systems we can describe are those that can be generated by a finite repertoire of procedures, i.e.\ those that have a grammar. 
Otherwise, the only `description' consists of listing all elements in the family---which can hardly be called a description at all.  
This does not only apply to families of physical systems, but also to families of complex systems \cite{Ho14c,Mi09,So00,Th18}.  
A simple counting argument shows that the vast majority of such families do not have a grammar \footnote{There are only countably many finite descriptions, and uncountably many infinite families. This follows e.g.\ by Cantor's diagonal argument. See, e.g., \cite{Ra19}.}.     
Yet, \emph{interesting} families are generally very special (see, e.g., \cite{De11e}) and tend to have a grammar. 
In fact, a finite description substantially contributes to making a family interesting. 

A powerful tool to model complex systems are classical spin Hamiltonians, 
where the `spin'  stands for any classical and discrete variable,  
and the `Hamiltonian' stands for any family of cost functions.  
The  paradigmatic spin Hamiltonian, the Ising model \cite{Is25}, 
has been applied far beyond its original purpose of modelling magnetism, for it has been used 
as a toy model of matter \cite{Am09a}, 
to model a gas \cite{Ch87}, 
in knot theory \cite{Ka01}, 
for artificial neural networks \cite{Ho82}, 
to model the size of canopy trees \cite{So00}, 
flocks of birds \cite{Bi12}, 
viruses as quasi-species \cite{An83,Ta92}, 
for protein folding \cite{Ba03} (together with its generalisation, the Potts model \cite{Ri20}), 
for economic opinions, urban segregation and language change   \cite{St08}, 
for random language models \cite{De19f}, 
social dynamics \cite{Ca09}  or
earthquakes \cite{Ji07}, 
and the US Supreme Court \cite{Le15c}.  
This wide applicability may be due to the recently discovered universality of the Ising model \cite{De16b}, 
meaning that it can simulate all other spin models \footnote{This notion of universality should not to be confused with the universality classes of spin models \cite{Wi75}. Universal spin models can simulate spin models in all universality classes, but beyond this observation it is unclear how the two notions of universality are related. We are currently developing a framework for universality to address this question \cite{St22a,St22b}; see also the outlook.}.

\begin{figure*}[th]\centering
\includegraphics[width=1\textwidth]{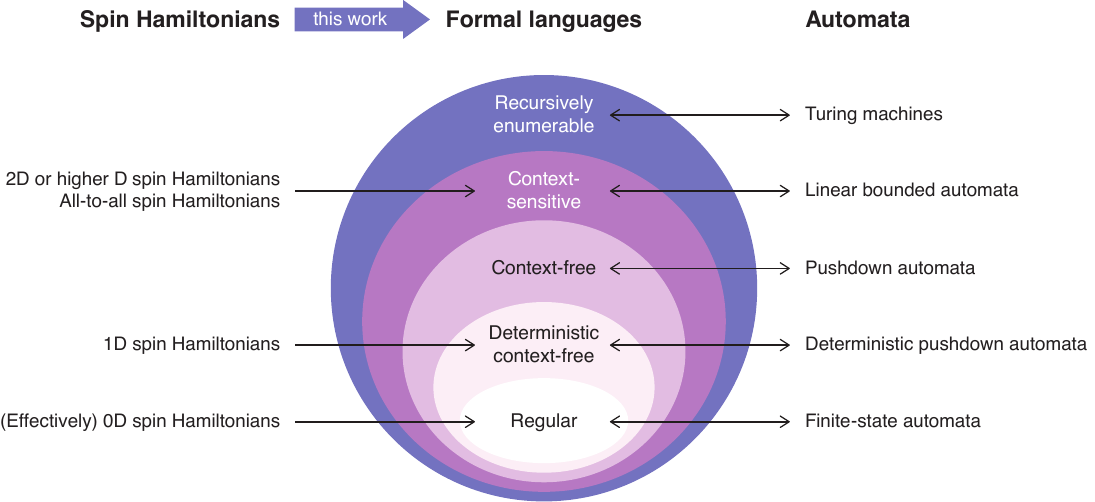}	
	\caption{{\textbf{Classifying the language of classical spin Hamiltonians in the Chomsky hierarchy.}} 
	(Left column) We cast classical spin Hamiltonians as formal languages, and classify them in the Chomsky hierarchy. 
	We show that the language of 
	0D spin Hamiltonians and effectively 0D spin Hamiltonians is regular, 
	the language of 1D spin Hamiltonians is deterministic context-free,
	and the language of $d$D spin Hamiltonians, with $d\geq2$ and all-to-all spin Hamiltonians is context-sensitive. 
The classification holds for any notion of locality with a bounded range, 
any (non-rectangular) lattice, 
any boundary conditions, 
and any symmetries.  
(Middle column)  The Chomsky hierarchy of formal languages, with the addition of deterministic context-free languages.   
	 (Right column)  Each level of the Chomsky hierarchy is associated to the class of automata recognising the corresponding languages. 
	}\label{fig:chomsky_hierarchy}
\end{figure*}

Intuitively, classical spin Hamiltonians are characterised by a local structure, given by the local Hamiltonians, which are the building blocks of the model. These building blocks play an analogous role to the grammar of a formal language. 
Can we describe the local interactions of a classical spin Hamiltonian as the grammar of a formal language?
If so, can we use the well-understood theory of formal languages to shed light on the complexity of the local structure of a spin Hamiltonian?
In other words, we are asking: 
\begin{quote}
\emph{What type of grammar does a spin Hamiltonian have?}
\end{quote}
In this work, we address this question by:
\begin{enumerate}[label=(\roman*)]
\item  establishing a new link between spin physics and formal languages, where we cast spin Hamiltonians as languages; and
\item  classifying the language of a spin Hamiltonian in the Chomsky hierarchy.  
\end{enumerate}
The latter renders a new complexity measure for spin Hamiltonians, which captures the hardness of the local structure.

More specifically, we consider a family of complex systems described by a classical spin Hamiltonian $H$, 
and cast $H$ as a formal language $L_H$. 
This language contains the set of all input--output pairs of $H$, 
that is, all pairs 
\begin{equation}
(x,H(x)) \quad \textrm{where $x$ is in the domain of $H$}. 
\end{equation}
We then classify $L_H$ in the Chomsky hierarchy of formal languages (\cref{fig:chomsky_hierarchy}) and prove that:  
\begin{enumerate}[label=(\roman*)]
\item 
The language of a 0D or effectively 0D spin Hamiltonian is regular;
\item  \label{ii} 
The language of a 1D spin Hamiltonian is deterministic context-free and not regular; and
\item \label{iii} 
The language of a $d$D for $d\geq2$ or all-to-all spin Hamiltonian is context-sensitive and not context-free. 
\end{enumerate}
We thus find that  essentially only the dimensionality of the lattice matters; 
the number of internal degrees of freedom,  
the notion of locality, 
the symmetries, 
or the boundary conditions are irrelevant for the classification, 
with the exception of effectively 0D spin Hamiltonians.  
 
This classification results in a new complexity measure of classical spin Hamiltonians that captures the hardness of the local structure of $H$, which is identified with the grammar of $L_H$.  
The usual complexity measure for spin Hamiltonians is given by the computational complexity of its ground state energy problem (GSE). 
We compare our new measure with the usual one, and find that they classify models differently. 
Specifically, for the Ising model, 
the new complexity  measure shows a different easy-to-hard threshold:  
for the GSE, 1D and 2D are  easy  and 3D is hard  
(in \textsf{P} and \textsf{NP}-complete \footnote{\textsf{P} and \textsf{NP} are the class of decision problems that can be recognised in polynomial time by a deterministic and non-deterministic Turing machine, respectively. A problem is \textsf{NP}-complete if it is in \textsf{NP}  and there is a polynomial-time reduction from any problem in \textsf{NP} to it. That the GSE of the 2D Ising without fields is in \textsf{P} follows from Onsager's solution, and the \textsf{NP}-completeness of the 3D Ising model was proven by Barahona \cite{Bar82}.}, respectively), 
whereas in our measure, 
1D is easy  and 2D and 3D are hard 
(deterministic context-free and context-sensitive, respectively). 

\begin{figure*}[th]\centering
\includegraphics[width=1\textwidth]{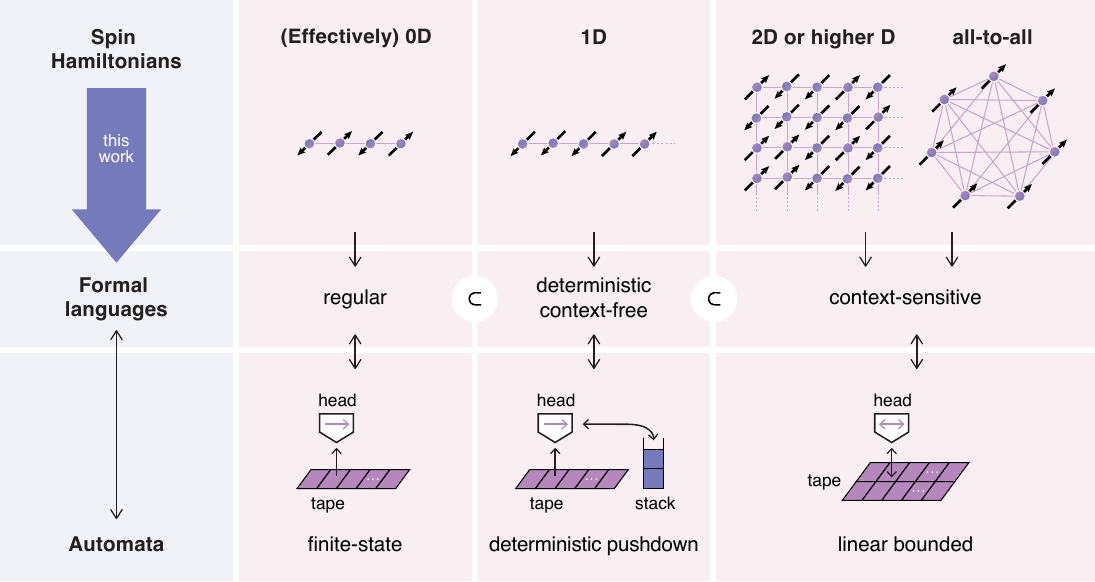}	
\caption{\textbf{Classical spin Hamiltonians, their languages, and the corresponding automata.} 
0D spin Hamiltonians are defined on a finite number of spins, 
and effectively 0D spin Hamiltonians are 1D spin Hamiltonians whose energy only depends on a finite number of spins. 
In both cases their language is regular, because it can be recognised by a finite-state automaton. 
The language of 1D spin Hamiltonians is deterministic context-free because it can be recognised by a deterministic pushdown automaton. 
The language of $d$D spin Hamiltonians for $d\geq2$ is context-sensitive, because it can be recognised by a linear bounded automaton. 
All-to-all spin Hamiltonians contain all possible $k$-body interactions, without any notion of lattice, and their language is also context-sensitive. 
}
\label{fig:main_result} 
\end{figure*}

In our complexity measure, 
once $L_H$ is fixed, its classification in the Chomsky hierarchy is unique, but $L_H$ depends on the casting of $H$ as a language. 
We partially characterise this freedom  by encoding the energy in binary, instead of unary, and   find that our classification (\cref{fig:chomsky_hierarchy}) only changes in that the language of 1D spin Hamiltonians becomes context-sensitive. 
In other words, a binary encoding can increase  the complexity, showcasing that the unary encoding is preferable. 

Finally, given a time evolution $U$, 
we define the language of the time evolution of a spin Hamiltonian, $L_U$,  as the set of all possible transitions, 
i.e.\ the set of all pairs $(x,U(x,t))$ such that $x$ is in the domain of $H$, for all times $t$. 
We classify $L_U$ in the Chomsky hierarchy and prove that $L_U$ is context-sensitive for all classes of spin Hamiltonians, 
except if $H$ is 0D, in which case it is regular, 
or if $H$ is effectively 0D or 1D, in either case with one spin value and fixed interactions, 
in which case it is deterministic context-free. 

This paper is structured as follows. 
We first define spin Hamiltonians and cast them as languages (\cref{sec:definitions}). 
Then we classify the languages of spin Hamiltonians in the Chomsky hierarchy (\cref{sec:mainresult}). 
We compare this complexity measure with the computational complexity of the ground state energy problem (\cref{sec:comparisonGSE}),  
investigate the dependence of our complexity measure on the encoding (\cref{sec:encoding}), and  
 the role of time (\cref{sec:time}). 
 Finally we conclude and present an outlook (\cref{sec:conclusion}). 
In \cref{ssec:eff0D} we characterise effectively 0D spin Hamiltonians, 
and in \cref{ssec:proofofmain,ssec:proofofcomputcompl,ssec:proofofbinary,ssec:proofoflangtime} we present the proofs of the results.

\section{Spin Hamiltonians and their language}\label{sec:definitions}

Here we define spin Hamiltonians (\cref{ssec:def}) and their associated language (\cref{ssec:language}). 

This section can be summarised as follows. 
We define a spin Hamiltonian $H$ as a map from spin configurations, intertwined with local interactions, to energies, for all system sizes. 
Every spin can take $q$ values (for fixed but arbitrary $q$),
and the local interaction specifies which local Hamiltonian is acting on which spins, 
where there are $p$ possible local Hamiltonians (fixed but arbitrary). 
Each local Hamiltonian acts on $k$ spins (fixed but arbitrary), and 
on which spins it acts will determine whether $H$ is a 1D, 2D or $d$D spin Hamiltonian, for any notion of a lattice, which need not be rectangular.    
For example, in 1D spin Hamiltonians the local Hamiltonians act on $k$ neighboring spins. 
We find a subclass of 1D spin Hamiltonians, called  effectively 0D spin Hamiltonians, 
whose energy is bounded (despite the fact that the number of spins is not), 
and depends only on a finite number of spins at each end of the chain---they are `holographic' 1D spin Hamiltonians. 
On the other hand, 0D spin Hamiltonians are defined as acting on a finite number of system sizes. 
Finally, we define all-to-all spin Hamiltonians as containing all possible $k$-body interactions, without any notion of lattice (\cref{fig:main_result}). 
The language of the spin Hamiltonian $L_H$ then contains all strings $e(x) u(H(x))$ for $x$ in the domain of $H$, where $e$ encodes the elements of the domain into strings, and $u$ is a unary encoding.

\subsection{Spin Hamiltonians\label{ssec:def}}

What is a spin Hamiltonian $H$? $H$ is a map from spin configurations to energies \emph{for all system sizes}. 
If $H$ were defined for a fixed system size, the corresponding language $L_H$ would be finite, and hence regular for trivial reasons. Ultimately, this is due to the fundamental distinction between finite and infinite made by computer science---any finite language is trivial, but infinite languages can be non-trivial in different ways.\footnote{Echoing Tolstoy's famous opening sentence, \emph{Happy families are all alike; every unhappy family is unhappy in its own way.}}

To be more precise, in this work a spin $s_i$ stands for a discrete classical spin, i.e.\ a variable that can take a finite number $q$ of values, which define the alphabet $\Sigma_q$.   
Now, $H$  maps  spin configurations  $s_1,\ldots, s_n$ to energies, for any system size $n$. 
Clearly, the maps for different sizes $n$ must be related---we formalise this relation 
with the notions of local Hamiltonian, local interaction and interaction structure (\cref{def:localinteraction}), 
which lead to a general definition of a spin Hamiltonian $H$ (\cref{def:spinhamiltonian}).  
We  later specialise to 0D, 1D and $d$D spin Hamiltonians, as well as all-to-all spin Hamiltonians, which we  then classify in the Chomsky hierarchy (\cref{thm:main}).

Let us now introduce the main actresses of \cref{def:localinteraction}. 
First, the \emph{local Hamiltonian} $h$ is a map from a finite number $k$ of spins to energies. 
This allows us to construct $H$ roughly as a sum of $n$ local Hamiltonians, for any $n$. 
There may be several local Hamiltonians (we call this finite number $p$)---this is for example the case in spin glasses, where the coupling strengths (and hence the local Hamiltonian) are drawn from a probability distribution. 
We include the information of which local Hamiltonian is acting as part of the input to $H$, that is, 
 $H$ is a map from $s_1,\alpha_1,\ldots, s_n,\alpha_n$ to energies, where $\alpha_j=1,\ldots, p=:[p]$ specifies the local Hamiltonian associated to spin $j$, namely $h_{\alpha_j}$. 
 
 Moreover, we take the energy to be integer-valued. 
 The reason is that computer science relies on entities with a finite description, so the rationals $\Q$ are a natural choice instead of the reals $\R$ (see \cref{sec:conclusion} for a discussion of other choices). 
Since there are finitely many local Hamiltonians, we can rescale all energies to the integers $\Z$.
 
Now, on which spins does the local Hamiltonian $h_{\alpha_j}$ act? 
This is described with a list of addresses $A_j$---each address is an integer number specifying the position of the corresponding spin with respect to $j$. For example, $A_j=(0,-1,1)$ refers to the spins $s_j, s_{j-1}, s_{j+1}$. 
The local Hamiltonian acts as $h_{\alpha_j} ((s_{j+l})_{l\in A_j})$, where 
 $(s_{j+l})_{l\in A_j}$ denotes the  tuple whose first element is $s_{j+l}$ where  $l$ is the first element in $A_j$, and so on. 
We call the pair $(\alpha_j,A_j)$ the \emph{local interaction} and denote it $I_j$.  

Not all possible interactions might be allowed in the input, 
and we formalise this idea with the \emph{interaction structure} $\mathcal{I}$:   
for a system size $n$ and  position $j$, the interaction structure at $n$ and $j$ (denoted $\mathcal{I}_{n,j}$) is an  element of the power set $\wp(\textrm{LocInt})$ of possible local interactions $\textrm{LocInt}$. 
Additionally, the interaction structure can be undefined for some system sizes; in this case, $H$ will be undefined for the corresponding system size. For example, for   0D spin Hamiltonians it will only be defined for a finite amount of system sizes, and for 2D spin Hamiltonians it will only be defined for multiples of a certain rectangular base lattice.

We also want to capture non-rectangular lattices, which can be seen as rectangular lattices with a non-trivial unit cell. 
That is, a unit cell with $t$ $r$-level spins can be described with a single $q$-level spin, where $q=r^t$, 
and other finite quantities are similarly enlarged. 
We will not resolve what happens inside the unit cell because it is irrelevant for our study---what matters is the scaling of $H$. 
So we will henceforth consider rectangular lattices, keeping in mind that each position can stand for a non-trivial unit cell.

The final element is a `missing spin symbol' $\missingspinsymbol$,  which can stand either for a spin value or for an address yielding that spin value. We use this symbol to describe  different boundary conditions in a unified way. 
Let us define 
\begin{equation}
\Z_\missingspinsymbol \coloneqq \Z \cup \{\missingspinsymbol\}, \quad 
\Sigma_{q\missingspinsymbol} \coloneqq \Sigma_q \cup \{\missingspinsymbol\}, 
\end{equation}
and its   $k$-fold Cartesian product as  $\Z_\missingspinsymbol^k$ and $\Sigma_{q\missingspinsymbol}^k$, respectively.  For any finite set $\Sigma$, we denote its Kleene star by $\Sigma^* = \cup_{k\geq 0} \Sigma^k$. 

We can now formally define the local Hamiltonian, a local interaction and the interaction structure. 

\begin{definition}
\label{def:localinteraction}
Let $q,k,p $ be positive integers. 
For every $\alpha \in [p]$, let the \emph{local Hamiltonian} be a function 
\begin{equation}
h_{\alpha}: \Sigma_{q \missingspinsymbol}^k \to \Z 
\end{equation}
A \emph{local interaction} $I$ is a pair
\begin{equation}
I =(\alpha,A) \in \mathrm{LocInt}: = [p]\times \Z_{\missingspinsymbol}^k   
\end{equation}
and the \emph{interaction structure} is a partial function 
		\begin{equation}
		\begin{split}
			\mathcal{I}: 
			\{(n,1), (n,2)\ldots (n,n) \mid  n \in \N\}
			& \rightharpoonup \wp \left( \mathrm{LocInt}
			 \right) \setminus \{\emptyset\} 
		\end{split}
	\end{equation}
	where for all $n$  we have that 
	$ \mathcal{I}_{n, j}$  is either defined for all $j \in [n]$ or for none.
\end{definition}

 We are now ready to define a spin Hamiltonian $H$.

\begin{definition}[Spin Hamiltonian]
\label{def:spinhamiltonian}
Let $q,k,p $ be positive integers.  
For all $\alpha\in [p]$ let $h_\alpha $ be a local Hamiltonian, and let $\mathcal{I}$ be an interaction structure (\cref{def:localinteraction}).
Define the domain $\mathcal{D}$ as
\begin{equation} \label{eq:domain}
\mathcal{D}=  
\bigcup_{n \geq 1,  \ \mathcal{I}_{n} \  \mathrm{defined}}
\left\{
	\left(
		s_1, I_1, \ldots , s_n, I_n
	\right) 
	  \mid s_j\in \Sigma_q ,I_j \in \mathcal{I}_{n,j} 
\right\}  
\end{equation}
A \emph{spin Hamiltonian} $H$ is given by a domain $\mathcal{D}$ and the map 
\begin{align} \label{eq:H}
\begin{split} 
H :\mathcal{D} &\to\Z \\
 (s_1,I_1,  \ldots , s_n, I_n)& \mapsto  
 \sum_{j=1}^{n}  h_{\alpha_j}( (s_{j+l})_{l \in A_j })
\end{split}
\end{align}
where $I_j = (\alpha_j, A_j)$, and where the symbol $\missingspinsymbol$ is used instead of $s_{j+l}$  if $l=\missingspinsymbol$ or if $j+l \notin [n]$. 
\end{definition}

\begin{figure*}[t!]\centering
	\includegraphics[width=.8\textwidth]{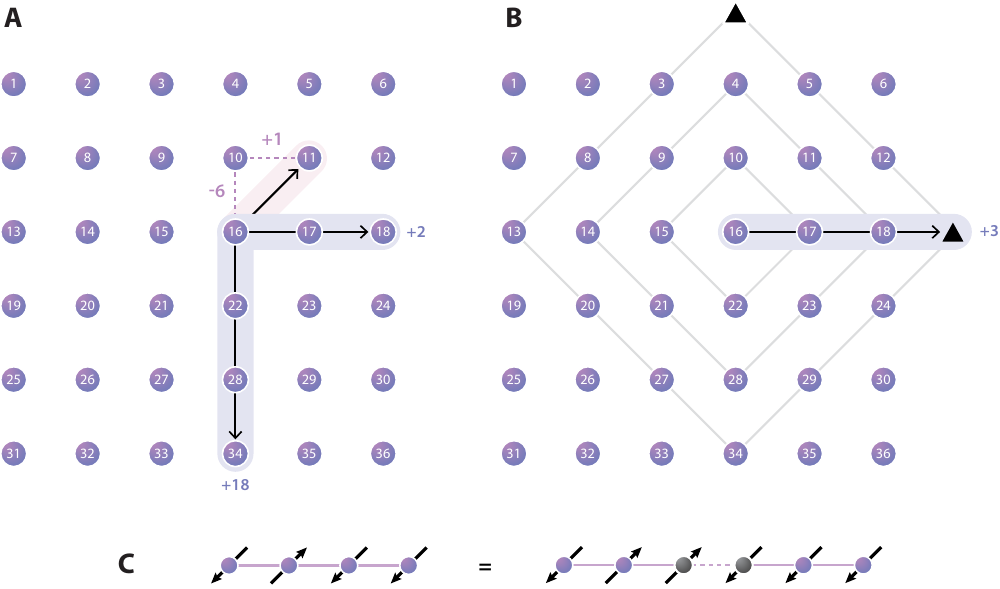}
		\caption{{\bf 2D spin Hamiltonians and effectively 0D spin Hamiltonians.} 
			{\bf A} In 2D spin Hamiltonians, 
			moves (vectors) on the lattice transform to moves along the string by subtracting the index of the source spin from the index  of the target spin. 
			The panel shows a 2D lattice with $n_1=n_2=6$ and that moving three rows down transforms to $+18$ in the string.  
			Moves in multiple dimensions (solid diagonal) can be decomposed into  moves along different dimensions (dashed lines). 
			{\bf B} Finite interaction ranges for $k=1,2,3$ around spin 16 are drawn in light gray. For $k=3$ some moves are not possible because they either point outside of the lattice or   lead to a new row, so the corresponding addresses are replaced with $\missingspinsymbol$. The arrow shows one such case, where moving $+3$ from spin $16$   yields spin $19$, which is not allowed.
			{\bf C}  In an effectively 0D spin Hamiltonian, the spin configuration on the left is equivalent to that on the right, as only the spins on the boundary of the chain contribute to the energy  (see \cref{pro:eff0D}).  
		}\label{fig:2d_enumeration}
\end{figure*}

We will cast any such spin Hamiltonian $H$ as a language $L_H$ (\cref{def:language}), 
and classify the language of certain classes of spin Hamiltonians in the Chomsky hierarchy (\cref{thm:main}). 
To this end, let us now define the classes of 0D, 1D, effectively 0D, and $d$D spin Hamiltonians, and finally the class of  all-to-all spin Hamiltonians.

%
\begin{definition}[0D spin Hamiltonian]\label{def:0D}
	A \emph{0D spin Hamiltonian} is a spin Hamiltonian (\cref{def:spinhamiltonian}) where the interaction structure $\mathcal{I}$ is only defined for a finite number of system sizes $n$.
\end{definition}

In other words,  the domain of a 0D spin Hamiltonian is finite, and hence so is its image. 
From a physical perspective, any such 0D spin Hamiltonian can be mapped to a Hamiltonian on a single (high-dimensional) spin---hence the name. 
From a computer science perspective, this definition highlights the fundamental distinction between finite and infinite. 
In thermodynamic terms, the spin Hamiltonian $H$ is an \emph{extensive} object, because the energy scales with the system size, and  our complexity measure (\cref{thm:main}) captures the complexity of the \emph{scaling} of $H$, i.e.\ the complexity of $H$ at the \emph{intensive} level---but in a 0D spin Hamiltonian there is no distinction between extensive and intensive properties.

%
We now turn to 1D spin Hamiltonians, where the spins are arranged on a 1D lattice, i.e.\ a chain of length $n \in \N$, for all $n$. 

\begin{definition}[1D spin Hamiltonian]\label{def:1D}
	A \emph{1D spin Hamiltonian} is a spin Hamiltonian (\cref{def:spinhamiltonian}) with spins arranged on a 1D lattice 
	and where the interaction structure is a total function such that for all $j\leq n$ and $n\in \N$, 
	\begin{equation}
(n,j) \mapsto \mathcal{I}_{n,j} \subset  [p] \times  \{-k, \ldots , k\}^k
	\end{equation}
	 is a constant function,  i.e.\ $\mathcal{I}_{n,j}$ is independent of both $n$ and $j$.
	We denote the resulting domain [Eq.\ \eqref{eq:domain}] by $\mathcal{D}_{\mathrm{1D}}$.
\end{definition}

In words, the interaction structure allows to choose one of $p$ local Hamiltonians, 
each of which acts on $k$ spins, and which can be at most $k$ positions far away from $j$, i.e.\ in the interval 
$\{-k,-k+1, \ldots ,k\}$. 
The key point is that $\mathcal{I}_{n,j}$ is constant, i.e.\ for every system size and every position in the chain, the set of interactions is the same.  
Recall that every position in the chain $j$ can in fact stand for a non-trivial unit cell; this  allows to describe 1D lattices with a finite periodicity. 
Note also that this definition includes the limiting case $k=1$, in which the spins do not interact with each other---although one can hardly call such a spin Hamiltonian 1D, the image of this $H$ can be unbounded, and thus behave like a non-trivial 1D spin Hamiltonian for the purposes of \cref{thm:main}.

While 1D spin Hamiltonians have an unbounded domain (as they are defined for all $n$) and an unbounded image in general, 
the latter is not the case for  a subset of 1D spin Hamiltonians, which effectively behaves like 0D spin Hamiltonians:


\begin{definition}[Effectively 0D spin Hamiltonian]\label{def:eff0D}
	An \emph{effectively 0D spin Hamiltonian} is a 1D spin Hamiltonian (\cref{def:1D}) with a bounded image.
\end{definition}

In \cref{ssec:eff0D}, we show that the energy of effectively 0D spin Hamiltonians only depends on a finite number of spins at each end of the string, as the energy of spins in-between cancels out. In this sense, they behave as holographic 1D spin Hamiltonians (\cref{fig:2d_enumeration}). 
Henceforth, we call spin Hamiltonians $1$D only if they are not effectively $0$D.

We now turn to spin Hamiltonians on $d$ dimensional lattices.  
Recall that non-rectangular lattices can be seen as rectangular lattices with non-trivial unit cells,  
so we will consider rectangular lattices. 
So let $R = (n_1, n_2, \ldots, n_d)\in \N^d$ be the side lengths of a rectangular lattice in $d\geq2$ dimensions. We enumerate the spins lexicographically along the dimensions, i.e. there are $n_1$ spins in a row, $n_1  n_2$ spins in a plane, $N_e=n_1   \cdots  n_e$ spins in an $e$-dimensional hyperplane and $n = n_1  \cdots n_d$ spins in total 
 (see \cref{fig:2d_enumeration}). 
With this notation we decompose any spin number as
\begin{equation}\label{eq:decompi}
	j = \sum_{i=0}^{d-1} a_i N_i + 1
\end{equation}
where $0 \leq a_i < n_{i+1}$, and where $N_0=1$. 

To have uniform families of Hamiltonians, we define a base  lattice $R_0 = ( \ell_1,  \ldots, \ell_d)$ and scaled versions thereof, 
\begin{equation} 
m R_0 \coloneqq (m \ell_1,  \ldots, m \ell_d) \quad \textrm{ for $m\in \N$}, 
\end{equation}
and denote the set of all such lattices by 
\begin{equation} \label{eq:curlyR}
\mathcal R \coloneqq \bigcup_{m \in \N} m R_0.
\end{equation}  
We now define a $d$D spin Hamiltonian  
based on two conditions: the first  ensures that the spins  be arranged on a lattice $R \in \mathcal{R}$, and the second  that the interaction range  be finite. 

\begin{definition}[$d$D spin Hamiltonian]\label{def:dD}
Let $d\geq 2$ be an integer.  
	Fix a base lattice $R_0 = (\ell_1,\ldots, \ell_d)$ and let $ \mathcal{R}$ be defined as in \eqref{eq:curlyR}. 
	A \emph{$d$D spin Hamiltonian} is a spin Hamiltonian  (\cref{def:spinhamiltonian}) with two properties:
	\begin{enumerate}
		\item $\mathcal{I}$ is defined for system size $n$  if and only if there is a lattice $R=(n_1,   \ldots, n_d) \in \mathcal{R}$ with $n= n_1 \cdots n_d$, and 
		\item There is a subset $V \subseteq [p] \times \left\{ \mathbf{v} \in \Z^d \mid \sum_i|v_i| \leq k \right\}^k$ such that the interaction structure for all $n, j$ is given by 
		\begin{equation}
		\mathcal{I}_{n,j} = \left\{(\alpha, (Z_{n,j}(\mathbf{v}^1), \ldots, Z_{n,j}(\mathbf{v}^k))) \mid (\alpha, (\mathbf{v}^1, \ldots, \mathbf{v}^k)) \in V \right\}
		\end{equation}
		where $Z$ is the function
		\begin{equation}\label{eq:dDZ}
			\begin{aligned}
				Z: \Z\times \Z \times \Z^d &\to \Z_\missingspinsymbol\\
					n,j, \mathbf{v}&\mapsto Z_{n,j}(\mathbf{v}) = 
					\begin{cases}
						\sum_{i = 0}^{d-1} v_i N_i & \forall \, i: \ 0 \leq a_i + v_i < n_{i+1} \\
						\missingspinsymbol & otherwise \\
					\end{cases}
			\end{aligned}
		\end{equation}
		where $a_i$ are given by \cref{eq:decompi}, $n_i$ are the side lengths and $N_i=n_1   \cdots  n_i$.
	\end{enumerate}
	We denote the resulting domain by $\mathcal{D}_{d\mathrm{D}}$.
\end{definition}


In words, $V$ specifies the local Hamiltonian (i.e.\ an element $\alpha  \in [p]$) and a list $\mathbf{v}$ of moves for each address, 
and $Z$ maps these lists to addresses. 
Specifically, the list $\mathbf{v} =( v_0, v_1, \ldots, v_{d-1})$ tells to move 
$v_0$ steps along the first dimension, 
$v_1$ steps along the second dimension, 
and so on  (\cref{fig:2d_enumeration}\textbf{A}), 
and  the total distance covered (i.e.\ the sum $\sum_i |v_i|$) is bounded by $k$ 
(\cref{fig:2d_enumeration}\textbf{B}). 
Since there is a list for each address and there are $k$ such addresses, the corresponding lists are denoted $\mathbf{v}^1, \ldots, \mathbf{v}^k$. 
The function $Z$ maps each list to a distance along the input---for example,  
$1 N_0 = 1$ moves to the next spin in the same row, 
$2 N_1$ moves two rows forward, and $(-1) N_2$ moves one plane backwards.
In addition, some moves are not possible because the spin lies at the edge. 
Consider for example spin $n_1$, the last spin in the first row; adding $1$ would bring us to the first spin in the second row which has a distance of $1+n_1$ from spin $n_1$, which is not allowed. That is, no move $v_j$ is allowed to jump to the next hyperplane (\cref{fig:2d_enumeration}\textbf{B}); 
the addresses corresponding to such a move are  replaced by the symbol $\missingspinsymbol$.
As a final remark, note that $d$D spin Hamiltonians have an unbounded domain, but are only defined for sizes $n$ corresponding to lattices in $\mathcal R$.

Finally, we  consider an even more general class of Hamiltonians, namely \emph{all-to-all} spin Hamiltonians, 
where the only ``restriction" is that each of the local Hamiltonians act on $k$ spins  
but these can be placed anywhere, i.e.\ there is no notion of physical locality with respect to a $d$ dimensional lattice. 
To exclude trivial cases we assume that $k\geq 2$, since for  $k=1 $ there are no interactions among spins and we recover a 1D spin Hamiltonian. 
We will show that its complexity is the same as that of $dD$ spin Hamiltonians (\cref{thm:main}), 
showcasing that this very loose notion of $k$-locality does not increase the complexity of $L_H$.   

To be precise, an all-to-all spin Hamiltonian contains   interactions among all subsets of $k$ spins from a total of $n$ spins, of which there are  $n \choose k$ (\cref{fig:main_result}). 
To formalise this, 
we associate to each spin $j$ 
all local interactions between this spin and any combination of $k-1$ spins with a larger index. 
There are 
\begin{equation}
c_j \coloneqq {{n-j}\choose{k-1}} 
\end{equation} 
many such combinations. 
We then replace the interaction structure (\cref{def:localinteraction})  with a set of lists of local interactions:
\begin{equation}
	\mathcal{J}_{n,j} = \left\{ \left( (\alpha_1, A_j^1),  \ldots, (\alpha_{c_j}, A_j^{c_j} ) \right) \mid \alpha_l \in [p] \right\}   
\end{equation}
where $A_j^l $ is a list of  addresses, namely
\begin{equation}\label{eq:order}
A_j^l = \left(0, a_1,  \ldots, a_{k-1} \right)
\end{equation}
where $0< a_1< \ldots<a_{k-1}\leq n-j$.  
These lists of addresses are ordered lexicographically, so that $A_j^l$ is the $l$th interaction associated to spin $j$. 
An all-to-all spin Hamiltonian only differs from a spin Hamiltonian  (\cref{def:spinhamiltonian})  in an  additional sum over interaction sets:

\begin{definition}[All-to-all spin Hamiltonian]\label{def:all-to-all}
Let $q,p$ be positive integers and let $k\geq2$.  Let the domain of an all-to-all spin Hamiltonian be 
\begin{equation}
	\mathcal{D}_{\mathrm{all}}= 
	\bigcup_{n \geq 1}
	\left\{
		\left(
			s_1, J_1, s_2, J_2, \ldots , s_n, J_n
		\right)
		\mid s_j\in \Sigma_q , J_j \in \mathcal{J}_{n,j}
	\right\} 
\end{equation}
An \emph{all-to-all spin Hamiltonian} $H$ is given by a domain $\mathcal{D}_{\mathrm{all}}$ and the map
\begin{align}
	\begin{split}
	H :\mathcal{D}_{\mathrm{all}} &\to\Z \\
	 (s_1, J_1,  \ldots, s_n, J_n)& \mapsto  
	 \sum_{j=1}^{n} \sum_{(\alpha, A) \in J_j} h_{\alpha}( (s_{j+l})_{l \in A })
	\end{split}
\end{align}
\end{definition}


So far, all spin Hamiltonians have been defined with open boundary conditions. 
To account for periodic boundary conditions in 1D spin Hamiltonians, in Eq.\ \eqref{eq:H} we redefine the sum as  $j+l = (j+l-1) \textrm{mod}(n) -1$ in $s_{j+l}$. 
For $d$D spin Hamiltonians, one can choose either periodic or open boundary conditions along any dimension, i.e.\ for any pair of opposite $(d-1)$-dimensional edges. 
Moreover, a rotation can be applied around the edges with periodic boundary conditions (resulting in \emph{quasi-periodic} boundary conditions), so that a $2$D lattice could become a torus or a M\"obius strip. 
In order to account for any such boundary conditions in $d$D spin Hamiltonians, 
we redefine the sum for the dimensions with periodic boundary conditions  similarly as above, 
and modify $Z$ [Eq.\ \eqref{eq:dDZ}] and $H$ accordingly.

\subsection{The language of a spin Hamiltonian \label{ssec:language}}

We now define the language of a spin Hamiltonian, where language stands for formal language throughout this work. 
The language $L_H$ of a spin Hamiltonian  $H$ is the set of all pairs $(x,H(x))$ such that $x\in \mathcal{D}$. 
More precisely, the words of the language consist of an encoding $e(x)$ of an element $x$ of the domain together with a unary encoding $u(H(x))$ of the energy.

Let the alphabet of auxiliary symbols be denoted as 
\begin{equation}
\Gamma \coloneqq  \{\lemo,\rimo,\square, \diamond, \missingspinsymbol,\addressdelimiter \},  
\end{equation}
where none of these symbols is in $\Sigma_q$. 
A \emph{unary encoding} of integers is a map
 $u:\Z \rightarrow \{\diamond, \square \}^*$, so that 
 \begin{equation}\label{eq:u}
u(n)=\begin{cases}
\diamond^{n} & n>0\\
\square^{|n|} & n<0\\
\epsilon & n=0\\
\end{cases}
\end{equation}
where $\epsilon$ is the empty string.  

Define the encoding of a \emph{local interaction} $I$ as the  map
\begin{equation}
	\begin{split}
		\gamma: \textrm{LocInt} &\to ([p] \cup \Gamma )^*\\
			(\alpha, A) = (\alpha, (a_1,  \ldots, a_k))&\mapsto \alpha \lemo u(a_1) \addressdelimiter  \ldots \addressdelimiter u(a_k) \rimo 
	\end{split}
\end{equation}
where $\lemo, \rimo$ are used as opening and closing  ``parentheses" in the string. In addition, we define $u( \missingspinsymbol)=\missingspinsymbol$. 
Finally the encoding of the domain of a spin Hamiltonian is defined as
\begin{equation} \label{eq:e}
	\begin{split}
		e: \mathcal{D} &\to \left(\Sigma_{q } \cup [p] \cup \Gamma  \right)^*\\
		(s_1, I_1 \ldots, s_n, I_n) &\mapsto s_1 \gamma(I_1)\ldots s_n \gamma(I_n)  
	\end{split}
\end{equation}

\begin{definition}[Language of a spin Hamiltonian]\label{def:language} 
Let $H$ be a spin Hamiltonian (\cref{def:spinhamiltonian}) with domain $\mathcal{D}$.  
Let $e$ be the encoding defined in \eqref{eq:e}, and $u$ the unary encoding defined in \eqref{eq:u}.
The \emph{language of $H$}, denoted $L_H$, is given by
\begin{equation}
\begin{split}
L_{H} = \{  e(x) \sun u(H(x))\newmoon  \mid x\in \mathcal{D}\} 
\end{split}
\end{equation}
\end{definition}

The symbol $\sun$ separates the input and  output of $H$,  whereas $\newmoon$ marks the end of the string, and none of these symbols is part of the image of $e$. 
For the language of an all-to-all spin Hamiltonian, $\gamma$ is replaced by an encoding  of lists of local interactions, i.e.\ 
 $J = (I_1, \ldots, I_m) $ is encoded as the string $\gamma(I_1) \ldots \gamma(I_m)$.

\section{Classification in the Chomsky hierarchy \label{sec:mainresult}}

We can classify the languages of spin Hamiltonians in the Chomsky hierarchy (\cref{fig:chomsky_hierarchy}). 

\begin{theorem}[Classification in the Chomsky hierarchy -- Main result]\label{thm:main}
    \begin{enumerate}[before=\leavevmode, label=(\roman*),ref=(\roman*)]
    \item \label{thm:0D} 
    If $H$ is a 0D spin Hamiltonian (\cref{def:0D}) or an effectively 0D spin Hamiltonian (\cref{def:eff0D}) then  $L_H$ (\cref{def:language}) is regular. 
    
    \item \label{thm:1D} 
    If $H$ is a 1D spin Hamiltonian (\cref{def:1D}) then $L_H$ is deterministic context-free and not regular.
    
    \item \label{thm:dD} 
    If $H$ is a $d$D spin Hamiltonian with $d \geq 2$ (\cref{def:dD}) 
    then $L_H$ is context-sensitive and not context-free. 
    
    \item \label{thm:all}  If $H$  is an all-to-all spin Hamiltonian (\cref{def:all-to-all}) then $L_H$ is context-sensitive and not context-free.
    \end{enumerate}
    In addition, \ref{thm:0D}, \ref{thm:1D} and  \ref{thm:dD} hold for open, periodic or quasi-periodic boundary conditions along any dimension. 
    \end{theorem}
        
    Recall that all statements hold for a non-trivial unit cell, so e.g.\ \ref{thm:dD} holds for non-rectangular lattices. 
    We prove this theorem in \cref{ssec:proofofmain}.

Describing  $H$ as a formal language $L_H$ is equivalent to representing  $H$ as an automaton. 
The local interactions of $H$ are reflected in the grammar of $L_H$, or the transition rules of the automaton. 
The  classification in the Chomsky hierarchy  gives the least complex automaton capable of recognising $L_H$, 
that is, that can tell whether a spin configuration $x$ and energy $E$ are `correct', i.e.\ $H(x) =E$ (\cref{fig:main_result}). 
We find that all languages $L_H$ can be recognised by automata which are weaker than Turing machines (despite our very general definition of a spin Hamiltonian), suggesting that universal spin models are weaker than universal Turing machines, since the 2D Ising model with fields is a universal spin model \cite{De16b}.

The classification relies on the following ideas: 
A 0D spin Hamiltonian gives rise to a finite language, which  is trivially regular, and in effectively 0D spin Hamiltonians  the energy only depends on a finite number of spins at each boundary, which also results in a regular language.  
A 1D spin Hamiltonian has a fixed and finite interaction range $k$, 
so a deterministic pushdown automaton (DPDA) can store the state of the previous $k$ spins in the head, 
and the energy accumulated so far in the stack. 
Once the entire spin configuration is read, the DPDA reads the energy in the tape and subtracts it from the stack, and accepts if and only if they coincide---rendering $L_H$  a deterministic context-free language. 
Finally, a 2D (or higher D) spin Hamiltonian has a fixed and finite interaction range on the lattice, but  when the spin configuration is cast as a string, there are interacting spins  arbitrarily far away from each other in the string. 
The jump in complexity from 1D to 2D is due to this unbounded distance. 
For this reason we need a linear bounded automaton (LBA), rendering $L_H$ a context-sensitive language. 
The same is true for $H$ with all-to-all $k$-body interactions (\cref{fig:main_result}).

\section{Comparison to the ground state energy problem \label{sec:comparisonGSE}}

We now compare the  complexity measure of \cref{thm:main} with the usual complexity measure, which is the computational complexity of the ground state problem (\cref{ssec:computcompl}).  
We compare them in detail for the most paradigmatic spin model, the Ising model  (\cref{ssec:Ising}).

\subsection{Comparison with computational complexity \label{ssec:computcompl}}

Usually, the complexity of a classical spin Hamiltonian is measured by  the running time of a Turing machine that searches for its ground state; more precisely, by the computational complexity of the ground state energy problem (GSE),  
which asks:  
\begin{quote}
Given a classical spin Hamiltonian defined on $n$ spins, $H_n$, and a number $K$, is there a spin configuration with energy below $K$? 
\end{quote}
The set of yes instances to this problem defines the formal language
$ L_{\textrm{GSE}}$, which contains the set of all pairs $(H_n,K)$ such that there exists $x $ such that  $ H_n(x) <K $. 

In contrast,  our language $L_H$ is the set of all input--output pairs of a  given $H$;  
in other words, it contains the set of yes-instances of the problem 
\begin{quote}
Given a spin configuration and a number, is this a valid spin configuration and does this number correspond to the energy of the spin configuration under $H$?
\end{quote}

Another difference consists of the fact that, in computational complexity theory, the automaton is fixed---it is  a Turing machine---, 
and one studies how much time or space this machine needs to solve a problem. 
In contrast, in our complexity measure,  the variable is the type of automaton, 
and we look for the simplest kind of automaton recognising the language. 
Note that we do not measure how long the automaton needs to recognise the language, as this time does not play any role in our complexity measure.

Despite these differences, we now compare the two complexity measures by 
studying the computational complexity of $L_H$. 
We find a different classification than that of \cref{thm:main}:  
the distinction between 1D and $d$D spin Hamiltonians  is washed out, 
and the effectively 0D spin Hamiltonian does not have the  same complexity as the 0D case. 
Explicitly, let $\textsf{P}$ be complexity class of all problems which can be solved in polynomial-time by a deterministic Turing machine, and $\textsf{LIN}$ the subclass of those problems that can be solved in linear time by the same machine. 

\begin{proposition}[Computational complexity of $L_H$] \label{pro:computcompl}
\begin{enumerate}[before=\leavevmode, label=(\roman*),ref=(\roman*)]

\item \label{pro:computcompl:0D} 
If $H$ is a 0D spin Hamiltonian (\cref{def:0D}) then  $L_H$ (\cref{def:language}) can be recognised in constant time.  

\item \label{pro:computcompl:eff0D} 
If $H$ is an effectively 0D spin Hamiltonian (\cref{def:eff0D}) then  $L_H \in \textsf{LIN}$.

\item \label{pro:computcompl:1D} 
If $H$ is a 1D spin Hamiltonian (\cref{def:1D}) then  $L_H \in \textsf{P}$. 

\item \label{pro:computcompl:dD} 
If $H$ is a $d$D spin Hamiltonian with $d \geq 2$ (\cref{def:dD}) 
then $L_H\in \textsf{P}$. 

\item \label{pro:computcompl:all} 
If $H$  is an all-to-all spin Hamiltonian (\cref{def:all-to-all}) then $L_H \in \textsf{P}$.
\end{enumerate} 
In addition, \ref{pro:computcompl:0D}, \ref{pro:computcompl:eff0D}, \ref{pro:computcompl:1D} and  \ref{pro:computcompl:dD} hold for open, periodic or quasi-periodic boundary conditions along any dimension.
\end{proposition}

The fact that $L_H \in \textsf{P}$ for all classes of spin Hamiltonians means that the graph of $H$ (i.e.\ the set of all input--output pairs of $H$) can be ``easily'' verified, and this essentially follows from the fact that $H$ is well-behaved. 
We prove this proposition in \cref{ssec:proofofcomputcompl}.

\subsection{The complexity of the Ising model \label{ssec:Ising}}

We now consider the most paradigmatic spin model, the Ising model, and compare our complexity measure to the computational complexity of its ground state energy problem. The Ising model  without fields is defined as 
\begin{equation}
H  (s_{1}, \ldots,s_n) =  \sum_{\langle i,j\rangle} J_{i,j}s_{i} s_j    
\end{equation}
where $\langle i,j\rangle$ denotes that $i,j$ are nearest neighbors in a given lattice, for all $n$.   
For GSE, the input to the problem is the set of coupling strengths $\{J_{i,j}\}$, the lattice and $K$, and $L_{\mathrm{GSE}}$ is given by the set of yes-instances to this problem.  
For our complexity measure, the elements of $L_H$ are given by 
the spin configuration, 
the local Hamiltonians (which encode the coupling strengths), 
the addresses (which encode the lattice) and the corresponding energy. 

One can  thus measure  the complexity of the Ising model with the computational complexity of recognising  $L_{\mathrm{GSE}}$  or with the classification  of $L_H$ in the Chomsky hierarchy. How do they compare? 
They exhibit a different threshold between `easy' and `hard': 
For GSE, the 1D and 2D Ising model are  easy  and 3D is hard  (in \textsf{P} and \textsf{NP}-complete\footnote{\textsf{P} and \textsf{NP} are the class of decision problems that can be recognised in polynomial time by a deterministic and non-deterministic Turing machine, respectively. A problem is \textsf{NP}-complete if it is in \textsf{NP}  and there is a polynomial-time reduction from any problem in \textsf{NP} to it.} \cite{Bar82}, respectively), 
whereas in our measure, 1D is easy  and 2D and 3D are  hard  (deterministic context-free and context-sensitive, respectively; see \cref{tab:Ising}). 
Note that the 1D Ising model without fields cannot be effectively 0D. 

\begin{table*}[th]\centering
	\begin{tabular}{|l|c|c|c|}
	\cline{2-4}
	\multicolumn{1}{c|}{} & 
	\multicolumn{1}{c|}{Complexity of $L_{H}$} & 
	\multicolumn{1}{c|}{Computational}  & 
	\multicolumn{1}{c|}{Computational}  
	\\
	\multicolumn{1}{c|}{} & 
	\multicolumn{1}{c|}{as a language} & 
	\multicolumn{1}{c|}{complexity of $L_{H}$}  & 
	\multicolumn{1}{c|}{complexity of $L_{\mathrm{GSE}}$ }  
	\\
	\hline 	
	1D Ising  &
	deterministic context-free & 
	in \textsf{P}  & 
	in \textsf{P}
	\\
	\hline
	2D Ising  & 
	context-sensitive& 
	in \textsf{P}  &
	in \textsf{P}
	\\
	\hline
	3D Ising  & 
	context-sensitive& 
	in \textsf{P}  &
	 \textsf{NP}-complete 
	\\
	\hline 
	\end{tabular}
	\caption{\textbf{Complexity of the Ising model.}  
	The Ising model without fields shows different easy-to-hard thresholds for the different  complexity measures. 
	If we measure the complexity of $L_H$ as a language, it is deterministic context-free for the 1D Ising model, and context-sensitive for the 2D and 3D Ising model (\cref{thm:main}). 
	The computational complexity of $L_H$ is in \textsf{P} for the 1D, 2D and 3D Ising model (\cref{pro:computcompl}).
		The computational complexity of recognising $L_{\textrm{GSE}}$  of the Ising model without fields is in  \textsf{P} if the model is defined in 1D or 2D, and  \textsf{NP}-complete if defined in 3D.
	}
	\label{tab:Ising}
\end{table*}

\section{The dependence on the encoding \label{sec:encoding}}

How unique is the complexity measure of $H$ provided by \cref{thm:main}? 
Once $L_H$ is fixed, its classification in the Chomsky hierarchy is unique, 
but $L_H$ depends on the encoding of $(x, H(x))$ as a string, which is not unique. 
For example, in the encoding presented in \cref{ssec:language}, 
a $d$D dimensional array of spins is cast as a 1D string,  
local interactions are part of the domain, and 
the energy is expressed in unary. 
Since we want to capture the complexity of $H$, 
it is meaningful to choose the encoding of $H$ that gives rise to the least complex $L_H$. 
Here we investigate the dependence of $L_H$ on the encoding of the energy. More specifically, we consider a \emph{binary} encoding of the energy, instead of a unary one, and prove that \cref{thm:main} only changes in that the language of 1D spin Hamiltonians is no longer deterministic context-free, but context-sensitive (\cref{pro:binary}).  In other words, a binary encoding of the energy \emph{increases} the complexity of the corresponding language in some cases---intuitively, because it is more difficult to process a binary than a unary number.
This showcases that the unary encoding of the energy is preferable. 

To make this precise, define the \emph{binary encoding} $b$ as the map 
\begin{equation}\label{eq:b}
\begin{split}
b:\Z&\rightarrow \{+,-\} \times \{0,1\}^*\\
n& \mapsto b(n) = \pm c_1\dots c_m
\end{split}
\end{equation}
where $m=\lfloor \log_2(|n|) \rfloor+1$ and $|n|=\sum_{l=1}^m c_l2^{l-1}$ and where the sign symbol is chosen according to the sign of $n$.
Note that we are slightly abusing of notation, as $\{0,1\}$ represent symbols of an alphabet in the first line and 
 integers in the second line of \eqref{eq:b}.

\begin{definition}[Language of a spin Hamiltonians with binary encoding] \label{def:lanbin}
Let $H$ be a spin Hamiltonian (\cref{def:spinhamiltonian}) with domain $\mathcal{D}$. 
Let $e$ be the encoding defined in \eqref{eq:e}, and $b$ the binary encoding defined in \eqref{eq:b}. 
The  \emph{language of $H$ with binary encoding}, denoted $L_{H}^{\mathrm{b}}$, is given by 
\begin{equation}
\begin{split}
	L_{H}^{\mathrm{b}}= \{  e(x) \sun b(H(x))\newmoon  \mid x\in \mathcal{D}\}
\end{split}
\end{equation}
\end{definition}

\begin{proposition}[Classification in the Chomsky hierarchy with binary encoding]
\label{pro:binary} 
Let $H$ be a spin Hamiltonian and $L_H^{\mathrm{b}}$  its language with binary encoding (\cref{def:lanbin}). 
Then \cref{thm:main} holds for $L_H^{\mathrm{b}}$ instead of $L_H$ except if $H$ is a 1D spin Hamiltonian, in which case $L_H^{\mathrm{b}}$ is context-sensitive and not context-free.  
\end{proposition}

We prove this proposition in \cref{ssec:proofofbinary}, and elaborate on further perspectives to characterise the freedom of $L_H$ in \cref{sec:conclusion}.

\section{The role of time \label{sec:time}}

So far, time has played no role in the construction of $L_H$ or in the classification. 
In order to investigate the role of time, we now cast the time evolution of a spin Hamiltonian as a formal language, and classify it in the Chomsky hierarchy. To this end, let us first introduce a general notion of a time evolution operator.

\begin{definition}[Time evolution operator]\label{def:timeevolution}
A time evolution operator $U$ is a function 
\begin{equation}
	\begin{split}
		U: \mathcal{D} \times T &\to \mathcal{D}\\
			(x,t) &\mapsto U(x,t)
	\end{split}
\end{equation}
where $T$ is an additive subgroup of $\R$.
In addition $U$ is a one parameter group, that is, for all $s,t \in T$
\begin{align}
	U(U(x,s),t) &= U(x,s+t) \\
	U(x,0) &= x
\end{align}
\end{definition}

Analogously to $L_H$, we now define the language $L_U$ which contains the set of possible transformations. 
More precisely, $L_U$ consists of the input and output configurations of $U$ where time is `integrated out', so that 
its elements are pairs of earlier--later. 
\begin{definition}[Language of the time evolution]\label{def:langtime}
	Let $H$ be a spin Hamiltonian (\cref{def:spinhamiltonian}) with domain $\mathcal{D}$, 
	and $U$ be a time evolution operator (\cref{def:timeevolution}). 
	The \emph{language of the time evolution $U$}  is defined as  
	\begin{equation}
	\label{eq:LU}
	L_U= \{e(x)  \sun  e(U(x,t)) \newmoon \mid x \in \mathcal{D}, t\in T_{\geq 0} \}
	\end{equation}
	where $e$ is the encoding given by \cref{eq:e}, and $T_{\geq 0}$ is the set of all nonnegative elements of $T$. 
\end{definition}

Note that  only the domain of the Hamiltonian appears in \eqref{eq:LU}, since the relation between $U$ and $H$ is so far unspecified. 
If $H$ were quantum mechanical, the time evolution operator would be given by $U_H = e^{-i H t}$; 
if $H$ were a Hamiltonian in classical mechanics, 
the time evolution would be given by Hamilton's equations
\begin{equation} 
\mathrm{d}q/\mathrm{d}t = \partial H / \partial p , \quad 
\mathrm{d}p/\mathrm{d}t  = - \partial H/ \partial q 
\end{equation} 
where $q$ and $p$ are the generalised coordinates and momenta, respectively \footnote{Obviously, this $q$ and $p$ are completely unrelated to the previous ones (denoting the number of spin values and number of local Hamiltonians, respectively).}.
But classical spin Hamiltonians are neither of these: 
they are neither quantum operators, 
nor classical Hamiltonians, because for the latter we would need to define a canonical pair of coordinates $q,p$ (and spin values can only provide one of those; 
the local interactions should not change under the time evolution, as they are additional parameters specifying the local structure). 
It follows that, if we only have a spin Hamiltonian $H$, the only choice for a time evolution is the trivial one $U = \mathrm{Id}_{\mathcal{D}}$, which gives rise to the language  
\begin{equation}\label{eq:langtime}
	L_U = \{ e(x) \sun e(x) \newmoon \mid x \in \mathcal{D} \}
\end{equation}

\begin{proposition}[Classification in the Chomsky hierarchy of $L_U$]\label{pro:langtime}
	Let $H$ be a spin Hamiltonian (\cref{def:spinhamiltonian}) 
	and $L_U$ be the language of the trivial time evolution given by \eqref{eq:langtime}.  
	\begin{enumerate}[before=\leavevmode, label=(\roman*),ref=(\roman*)]
	\item \label{pro:langtime:0D}
	If $H$ is a 0D spin Hamiltonian then $L_U$ is regular. 
	\item \label{pro:langtime:special}
	If $H$ is an effectively 0D spin Hamiltonian or a 1D spin Hamiltonian, in both cases with 
	  one spin symbol ($q=1$) and fixed interactions ($I_1= \ldots =I_n$ for all $n$), 
	 then $L_U$ is deterministic context-free and not regular. 
	\item \label{pro:langtime:other}
	If $H$ is an effectively 0D spin Hamiltonian or a 1D spin Hamiltonian with more than one spin symbol or unfixed interactions, or if $H$ is a $d$D spin Hamiltonian with $d\geq 2$ or an all-to-all spin Hamiltonian, then 
	$L_U$ is context-sensitive and not context-free. 
		\end{enumerate}
\end{proposition}

We prove this proposition in \cref{ssec:proofoflangtime}. 
This result can be intuitively understood as follows. 
In the 0D case (statement \ref{pro:langtime:0D}), the language is finite and hence regular. 
In the general case (statement \ref{pro:langtime:other}),  
the language is of the form $\{ww \mid w \in \Sigma^* \}$ where $w$ are words over some alphabet $\Sigma$, 
which is context-sensitive.  
In case \ref{pro:langtime:special}, the language is of the form $\{a^n \sun a^n \mid n\in \N\}$, where $a$ is a fixed word, which can be recognised by a DPDA.

We thus find that $L_U$ is generally more complex than $L_H$. In words, this owes to the fact that the automaton recognising $L_U$ needs to compare elements of the domain, which is more difficult than comparing an element of the domain with an energy in unary, as in $L_H$.

\cref{pro:langtime}  captures the complexity of transitions of states for a spin Hamiltonian with a trivial time evolution, and opens the door to a comparative study of the complexity of the language $L_U$ [Eq.\ \eqref{eq:LU}] for non-trivial time evolutions.

\section{Conclusions and Outlook \label{sec:conclusion}}

In this work, we have established a new relation between spin physics and theoretical computer science 
 by casting  classical spin Hamiltonians as formal languages.
Specifically, we have provided a general  definition of spin Hamiltonians (\cref{def:spinhamiltonian}) 
and cast them as  languages (\cref{def:language}). 
We have leveraged this relation  to classify the language of a spin Hamiltonian in the Chomsky hierarchy (\cref{thm:main}), and thereby provided  a new complexity measure of spin Hamiltonians. 
We have compared our complexity measure to the computational complexity of the ground state energy problem, 
and found a different classification (\cref{pro:computcompl}). 
This leads, in particular,  to different easy-to-hard thresholds for the Ising model (\cref{tab:Ising}). 
We have also investigated the freedom of casting $H$ as a language $L_H$ by encoding the energy in binary, and have shown that  this can increase the complexity  (\cref{pro:binary}). 
We have also defined the language of the time evolution of a spin Hamiltonian $L_U$, and classified it in the Chomsky hierarchy for a trivial time evolution (\cref{pro:langtime}).

This work sets the stage for many further investigations. 
One has to do with characterising the freedom of the map from $H$ to $L_H$ more thoroughly, 
as a complexity measure would choose the least complex $L_H$. 
While we have partially characterised this freedom by considering the binary encoding of the energy, 
a full characterisation of this freedom and a minimisation of the resulting set would strengthen the complexity measure. 
While the full characterisation may be within reach, 
the minimisation may be undecidable due to Rice's Theorem \cite{Ko97}. 
The situation is reminiscent  of algorithmic information theory: 
just as the Kolmogorov complexity mildly depends on the choice of universal Turing machine (see e.g.\ \cite{Hu10b}), 
the complexity of $L_H$ as a language depends on the map from $H$ to $L_H$;  
and  just as the Kolmogorov complexity is uncomputable, determining the least complex $L_H$ may be uncomputable too.

A further question concerns the scope of $H$ itself, i.e.\ what \emph{is} a spin Hamiltonian? 
A crucial property  for our relation is the discreteness of the spins---continuous variables (as considered in \cite{De16b}) ought to be discretised in order to connect them to formal languages. 
Another property is the energy: instead of the integers we could have chosen the algebraic reals, as they allow for a finite description. 
However, it is the domain of $H$ the part that is least obvious to define, as it relies on a total numbering of the spins as well as the notion of   local interaction. 
This can be circumvented by casting spin Hamiltonians as languages \emph{directly at the intensive level}, that is, casting the local Hamiltonian $h$ as the grammar $G_h$ of the language---for example, the local Hamiltonian of the 1D and 2D Ising model can be cast as a context-free and context-sensitive grammar, respectively \cite{Re21c}. This perspective results in a less cumbersome definition of the domain of $H$, but a less thorough classification than that of \cref{thm:main}. 
Alternatively, spin Hamiltonians could be defined on families of graphs, and  they could be characterised with graph grammars \cite{Ro97}.  


Another  investigation involves  identifying the boundaries of the Chomsky hierarchy in the features of $H$ and $L_H$. 
For example, what is the minimal addition to a 1D spin Hamiltonian so that its $L_H$  is non-deterministic context-free or context-sensitive? 
Preliminary results suggest that 1D lattices with a few long-range connections trigger this jump in complexity---these graphs could be seen as small-world graphs \cite{Wa98}, suggesting a  link to complex networks. 
This boundary can be fully characterised for the Ising model \cite{Re21c}. 
On the other hand, 
grammatical sentences of a natural language form a mildly context-sensitive language \cite{Fr11b}, 
suggesting that when these sentences are seen as spin models, these are more connected than 1D but less than 2D lattices.

Another question is: What is the relation between our complexity measure and the existence of phase transitions in a spin model? This relation is not transparent for the computational complexity of the ground state energy problem (GSE), as GSE of both the 1D and 2D Ising model without fields is in \textsf{P}, 
but the second one does have a  phase transition whereas the first one does not. Instead, the complexity of GSE is related to the existence of an algorithm to solve it. 
It is intriguing that the complexity of $L_H$ as a language for the Ising model does correlate with the existence of a phase transition (\cref{tab:Ising}). 

This raises the question of symmetries, 
which seem to  play a different role in our complexity measure than in statistical mechanics. 
On the one hand,  the definition of 1D and 2D spin Hamiltonians crucially relies on the ``external" symmetry of the lattice, which are the decisive factor for their complexity. 
This  external symmetry is a property of the domain $\mathcal{D}$ and the corresponding energy---for example, it says that if an element of the domain and the  energy are scaled in a certain way, the result will also be an element of the language. 
This transformation is captured by a \emph{grammar}---and we find that all languages $L_H$ in this work have a grammar (regular, deterministic context-free or context-sensitive). 
On the other hand, ``internal" symmetries (i.e.\ symmetries of the set of local interactions) do not seem to play  a role,  because the local Hamiltonian can be given as  part of the input, with the  exception of the class of effectively 0D spin Hamiltonians.

One question which  motivated this  work was the wish to rigorously compare universal spin models \cite{De16b,Ko20} and universal Turing machines \cite{De20d}. 
One route to comparing these two notions of universality is by putting the two objects of study---spin models and languages---at the same level, particularly at the level of languages. This  is what we have achieved in casting $H$ as $L_H$. 
The classification of $L_H$ in the Chomsky hierarchy suggests that universal spin models are \emph{weaker} than universal Turing machines, since  the 2D Ising model with fields is a universal spin model \cite{De16b} but its language is context-sensitive (\cref{thm:main}), and thus can be recognised by a linear bounded automaton, which is weaker than a Turing machine (\cref{fig:chomsky_hierarchy}).  
Yet, the 2D Ising model without fields is not universal, but its language is also context-sensitive in our classification. 
This indicates that we also need to translate the relevant transformation for spin models (called \emph{simulation} \cite{De16b}) to  a relation among the corresponding languages. 
This relation will be weaker than that provided by computable reductions \cite{Ko97}. 
How will the universality of spin models manifest itself in the world of languages? 
And, conversely, does the universality of Turing machines have no implications for the universality (as defined in \cite{De16b}) of spin models?

A different route to addressing these questions is to devise a framework for universality that expresses its relevant features abstractly, 
and which contains universal spin models and universal Turing machines (as well as many others) as instances. 
That framework would allow us to distinguish types of universality. 
This is what we are attempting to do in our categorical framework for universality \cite{St22a}. 

Understanding the relation between these universalities will also allow to explore the scope of undecidability in spin models,  and by extension in physics and complex systems. 
Undecidability plays a central role in computer science---for example, every non-trivial property  of a recursively enumerable language is undecidable by Rice's Theorem \cite{Ko97}. 
The relation between $H$ and $L_H$ is one step toward making a similar statement for spin models (see \cite{Gu08} for a different approach).   
Finally, quantum spin models could be related to notions of quantum computation, and   their notions of universality could be compared \cite{Cu17,Pi20}.



\noindent \textbf{Acknowledgements:}  
We thank our friends and colleagues Bernat Corominas-Murtra, Ricard Sol\'e, Tobias Reinhart and Tim Netzer for discussions, M$\bar{\mathrm{a}}$ris Ozols for  comments regarding the importance of the rescaling of $h$ for the classification of the language, and Thomas Tappeiner for help in the extension of the results to non-rectangular lattices.

\noindent \textbf{Funding:} 
This work was supported by the START Prize of the Austrian Science Fund (FWF) (project Y 1261-N).

\noindent \textbf{Author contributions:} 
DD and GDLC proved a preliminary version of the results of this paper. SS unified and generalised the results, added new results, and wrote the current manuscript together with GDLC.


\printbibliography[title={References}]


\newpage 
\appendix

\begin{center}
\section*{ Supplementary Material}
\end{center}

\bigskip

\section{Effectively 0D spin Hamiltonians 
\label{ssec:eff0D}}

Here we characterise effectively 0D spin Hamiltonians, which are  defined as spin Hamiltonians with a bounded image (\cref{def:eff0D}). 
For any $x \in \mathcal{D}$, we define its ``cardinality" as $|x| = n$ (instead of $2n$), and this denotes the number of spins.

\begin{proposition}[Characterisation of effectively 0D spin Hamiltonians]\label{pro:eff0D}
	Let $H$ be a 1D spin Hamiltonian. The following are equivalent: 
	\begin{enumerate}
	\item $H$ is an effectively 0D spin Hamiltonian (\cref{def:eff0D}). 
	\item $H$   only depends on the first and last $ m \leq k$ spins and local interactions. 
	That is, there is an $ m \leq k$ such that for every configuration $ (s_1,I_1, \ldots, s_n, I_n)$ with $n > 2  m$
	\begin{equation}
	\begin{split}
		H(s_1,I_1, \ldots, s_n, I_n) =  \\
		H(s_1,I_1, \ldots, s_{  m}, I_{ m}, s_{n-  m}, I_{n- m}, \ldots, s_n, I_n)  .
		\end{split}
	\end{equation}
	\end{enumerate}
\end{proposition}

In words, the energy of an effectively 0D spin Hamiltonian depends solely on $k$ spins at each end of the chain (\cref{fig:2d_enumeration}). 
They thus behave as `holographic' 1D spin models, as the interactions in the middle of the spin chain cancel out, and only the behaviour at the boundary matters. This results, in particular, in a bounded image of $H$. 

An example of an effectively 0D spin Hamiltonian is given by $\Sigma_q = \{0, 1\}$, the local interactions $\mathcal{I}_{n,j} = \{(1,(0,1))\}$,  
and  the local Hamiltonian   
$h_1(s_1, s_2)= 1 $ or $-1$ if $s_1s_2=01$ or 10, respectively, and 0 otherwise.  
The key property is that only non-equal adjacent spins  contribute to the energy, but there cannot be two $01$ spin pairs without a $10$ between them and vice versa. It follows that the energy is bounded (it is $-1,0$ or $1$), and can be determined by just checking the first and last spin.

\begin{proof}[Proof of \cref{pro:eff0D}]
	\emph{2.} $\Rightarrow$ \emph{1.} holds because there are only a finite number of configurations for a finite number of spins. The size of the image is bounded by the number of configurations and hence finite.
	
	\emph{1.} $\Rightarrow$ \emph{2.} 
	Let $H$ be a 1D spin Hamiltonian and let $x = (s_1, I_1, \ldots, s_n, I_n)$ be a configuration with $n > 2k$ spins.
	As the maximal interaction range is $k$ (by \cref{def:0D}), 
	there is at least one local interaction in $x$ which depends only on spins within $x$, i.e.\ which does not address spins outside $x$. Choose one such local interaction and denote its position by $\sigma$.

	We start by showing that the local interactions are fixed, i.e.\ that the energy is independent of the choice of local interactions. This is done by contradiction: Assume that there are at least two local interactions $I$ and $I'$ which lead to different energies for at least one spin configuration. 
	Choose $x$ such that the local interaction at position $\sigma$ is $I$ and define $x'$ as the same configuration as $x$ but where $I$ is replaced with $I'$ at position $\sigma$.
	As the two local interactions are different there is at least one such pair $x, x'$ with different energies, 
		$H(x) = H(x') + \Delta $, 
	where $|\Delta| > 0$. 
	As the energy of the local interaction at position $\sigma$ only depends on spin configurations defined within $x$ ($x'$), taking $l$ copies of $x$ and $x'$ will yield an energy difference of $l  \Delta$, i.e. 
$		H(x^l) = H(x'^l) + l \Delta$. 
	In the limit of $l \to \infty$ the total energy difference $l\Delta$ diverges, which contradicts the assumption that the image of $H$ is bounded. 

	We now show the contrapositive of 	\emph{1.} $\Rightarrow$ \emph{2.}
	So assume $H$ does not have property \emph{2.}, i.e.\  for every $m\leq k$ there exist configurations $x = x_1x_2, |x_1|=|x_2|=   m$ and a longer one $x'=x_1x_3x_2, |x_3|>0$  such that
	$ H(x)-H(x')=\Delta$ with $\Delta\neq 0$. 
	We want to show that the image of $H$ is unbounded. 
	First, choose $  m = k$ and find two such configurations $x$ and $x'$. 
	Consider now the $l$-fold repetition of the configurations, $x^l, x'^l, l\geq 1$. 
	As $x$ and $x'$ coincide on the first and last $m = k$ spins, the energy of the additional local interactions in $x_3$ will only depend on (i.e. address) spins in $x'$. Hence, taking $l$ times $x$ and $x'$ will also multiply the energy difference by $l$:
		$H(x^l)-H(x'^l)= l\Delta$. 
	As above, taking the limit $l \to \infty$ leads to a diverging energy difference,  and hence an unbounded image of $H$.
\end{proof}	

In the previous proof we have shown a useful fact: 

\begin{corollary}\label{cor:eff0D}
If $H$ is an effectively 0D spin Hamiltonian, then it has  fixed local interactions.  
	That is, for all $n,j$ and all spin configurations, all configurations of local interactions lead to the same energy.
\end{corollary}

Are there \emph{effectively 1D spin Hamiltonians} as a subset of 2D spin Hamiltonians, or more generally, 
effectively $(d-1)$D spin Hamiltonians as a subset of $d$D spin Hamiltonians? 
Namely, does this holographic property appear in higher dimensions, too? 
Since we are ultimately interested in the complexity of the corresponding language, 
and the languages of $2D$ and higher D spin Hamiltonians are in the same class (\cref{thm:main}), 
the only relevant case is that of effectively 1D spin Hamiltonians. 
Recognising the domain of a 2D spin Hamiltonian, however, already requires a  linear bounded automaton, regardless of the choice of local interactions. That is, an effectively 1D spin Hamiltonian needs to be a special case of a 2D spin Hamiltonian, and the language of the latter is context-sensitive for all cases.

\section{Proof of \cref{thm:main} (Classification in the Chomsky hierarchy) 
\label{ssec:proofofmain}}

The proof of \cref{thm:main} is structured following the items of the result:
we prove \ref{thm:0D} in \cref{sssec:proofofmain(i)},
\ref{thm:1D} in \cref{sssec:proofofmain(ii)},
\ref{thm:dD} in \cref{sssec:proofofmain(iii)},
\ref{thm:all} in \cref{sssec:proofofmain(iv)}, 
and the result with other boundary conditions in 
\cref{sssec:proofofmainother}.

\subsection{Proof of \cref{thm:main} \ref{thm:0D} \label{sssec:proofofmain(i)}}

The language of a 0D spin Hamiltonian is regular because $H$ has a finite domain (any finite set of strings defines a regular language \cite{Ko97}). 

To prove the statement for effectively 0D spin Hamiltonians, we provide the following deterministic finite-state automaton (DFA)  recognising the language. 
By \cref{pro:eff0D}, the values of the first and last $k$ spins suffice to calculate the energy, and the choices of local interactions will not affect the energy (\cref{cor:eff0D}), so 
the Hamiltonian can be written as a function of $2k$ spins,  $H_{2k}(s_1,\ldots,s_k, s_{n-k}, \ldots, s_n)$.
The automaton recognising $L_H$ needs to check the energy using the values of the first and last $k$ spins,  
and that the local interactions are elements of the set of allowed local interactions $\mathcal{I}_{0D}$ --- which is independent of both $n$ and $j$. 

To define the states of the DFA, 
let $E_{\min} = \min\left(H(\mathcal{D}) \cup \{0\}\right)$ where $H(\mathcal{D}) $ 
is the image of $H$,  and similarly for $E_{\max}$. 
The states of DFA are given as tuples 
\begin{equation}
\begin{split}
	\Large{\{}(F, L, E, i) \mid &
	\  F,L \in \cup_{j=0}^k \Sigma_{q \missingspinsymbol}^j,  \\
	& E \in \{E_{\min}, \ldots, E_{\max}\} ,\ i \in \Sigma_I  \Large{\}}
	\end{split}
\end{equation}
together with the states
\begin{equation}
	\{\textrm{Accept}, \textrm{Reject}\}
\end{equation}
where $F, L$ contain the first and last $k$ spins, respectively, 
$E$ is used to store the energy, and $i$  is used to validate the local interactions, where we have defined
\begin{equation}\label{eq:SigmaI}
\Sigma_I \coloneqq \bigcup_{j=0}^{M} ([p] \cup \{\lemo,\rimo,\addressdelimiter ,\square, \diamond\})^j
\end{equation} 
with 
\begin{equation} 
 M \coloneqq \max_{I \in \mathcal{I}_{0D}} |\gamma (I)|.
 \end{equation}
The initial state of the automaton is $(F=\epsilon, L=\epsilon, E=0, i=\epsilon)$ and it starts at the first spin of the input, $s_1$.
Denote the symbol currently being read by $C$.
The transitions are then given by \cref{alg:DFA_part1,alg:DFA_part2},  
where the convention is used that the string is rejected if a transition is not possible, e.g.\ if the desired target state does not exist.
In words, \cref{alg:DFA_part1} validates the input until the symbol $\sun$, and saves the first and last $k$ spins to $F$ and $L$, respectively. 
\cref{alg:DFA_part2} calculates and validates the energy (on \cref{alg:DFA_calcE}); 
more precisely, for each pair $(F,L)$ there is a transition to some state $E \in \{E_{\min}, \ldots, E_{\max}\}$.
Finally, the energy given on the input is subtracted from $E$ (the energy saved in the head),  and the DFA accepts only if $E=0$ when reaching $\newmoon$.

\begin{figure}[htb]
\begin{algorithm}[H]
\caption{Transitions of the DFA, Part 1: Loading the spins into the memory}\label{alg:DFA_part1}
\begin{algorithmic}[1]
\State Assert that $C \in \Sigma_{q \missingspinsymbol}$
\While{$C \neq \sun$}
	\Comment{Append the spin to the state}
	\If{$|F|\leq k$} \Comment{First fill $F$}
		\State Append $F \gets FC$
	\ElsIf{$|L|\leq k$} \Comment{Then fill $L$}
		\State Append $L \gets LC$
	\Else \Comment{Keep the last $k$ spins in $L$}
		\State Remove the leftmost spin in $L$
		\State Append $L \gets LC$
	\EndIf
	\State Move to the next symbol
	\While{$C \notin \Sigma_{q \missingspinsymbol}$}\Comment{Read local interaction into $i$}
		\State Update $i \gets iC$ 
		\State Move to the next symbol
	\EndWhile
	\State Assert that $i \in \mathcal{I}_{0D}$
	\State Reset $i \gets \epsilon$
\EndWhile
\algstore{DFA}
\end{algorithmic}
\end{algorithm}
\end{figure}

\begin{figure}[htb]
\begin{algorithm}[H]
\caption{Transitions of the DFA, Part 2: Calculating and validating the energy \label{alg:DFA_part2}}

\begin{algorithmic}[1]
\algrestore{DFA}
\State Set $E \gets H_{2k}(FL)$ where $FL$ is $F$  concatenated with $L$ \label{alg:DFA_calcE}
\State Move to the next symbol
\If{$C = \diamond$} \Comment{Save the first symbol of $u(E)$}
	\State $i \gets \diamond$
\ElsIf{$C = \square $}
	\State $i \gets \square$
\ElsIf{$C = \newmoon$ and $E=0$}
	\State Accept
\Else
	\State Reject
\EndIf

\While{$C \neq \newmoon$}
	\State Assert $C = i$ \Comment{Check for a valid unary encoding}
	\If{$i = \diamond$}
		\State Set $E \gets E-1$ \Comment{Subtract the input from $E$}
	\Else
		\State Set $E \gets E+1$ \Comment{Add the input from $E$}
	\EndIf
	\State Move to the next symbol
\EndWhile
\If{$E=0$}
	\State Accept
\Else
	\State Reject
\EndIf
\end{algorithmic}
\end{algorithm}
\end{figure}

\subsection{Proof of \cref{thm:main} \ref{thm:1D} \label{sssec:proofofmain(ii)}}

\paragraph{The language of a 1D spin Hamiltonian is deterministic context-free}

We provide a deterministic pushdown automaton (DPDA) recognising the language. 
Similarly to the effectively 0D case the state of the machine is used to save the last read $2k+1$ spins and the last $k$ local interactions. Recall that in the 1D case $\mathcal{I}_{n,j}$ is independent of both $j$ and $n$, and define $\mathcal{I}_{1D} = \mathcal{I}_{n,j}$ for any $n,j$. 

To save partially read local interactions we use $\Sigma_I$ [Eq.\ \eqref{eq:SigmaI}], and  
let the states of the automaton be 
\begin{equation} \label{eq:headDPDA}
\begin{split}
	\left\{(L_s, L_I, i) \mid \ L_s \in \cup_{j=0}^{2k+1} \Sigma_{q  \missingspinsymbol}^j, \ L_I \in \cup_{j=0}^{k} \mathcal{I}_{1D}^{j}, \ i \in \Sigma_I \right\}\\ 
	\cup \{\textrm{Accept}, \textrm{Reject}\}. 
	\end{split}
\end{equation}
Note that there are only finitely many states,  as $\mathcal{I}_{1D}$ is a finite set. 
The stack symbols are $\{Z, \diamond, \square\}$ where $Z$ is the initial stack symbol, and the initial state is given by $(L_s=\epsilon, L_I= \emptyset, i=\emptyset )$.

The automaton starts at the first spin of the input, $s_1$.
The transitions are then given by \cref{alg:DPDA_part1,alg:DPDA_part2}, where, as before, 
the  symbol currently read is denoted $C$, and if any transition is not possible the automaton rejects the string.
Two lines of  \cref{alg:DPDA_part1} need further explanation. 
On \cref{alg:DPDA_energy1} the energy of a single local interaction is calculated, namely that of the last $k$ steps. 
This ensures that the $2k+1$ saved spins are exactly the spins with distance $\leq k$ from the local interaction. 
The function mapping $(L_s, I)$ to some energy $E$ is hardwired in the machine, as the domain of this function is finite. 
On \cref{alg:DPDA_energy2} we encounter a similar situation, with the only difference that the function is now from $(L_s, L_I)$ to some energy $E$---the domain is also finite and hence the function can be hardwired.

When pushing the energy onto the stack, the energy already on the stack may be of a different sign than the energy to be pushed ($E$). In this case, the automaton pops symbols from the stack until $E$ symbols have been removed, or until $Z$ is on top of the stack. 
In the latter case, the automaton pushes $u(E-P)$, where $P$ is the number of symbols removed.
This will always lead to a valid unary encoding of some energy on the stack.
Note that the number of simultaneous push/pops is bounded,  as the energy of a local Hamiltonian is bounded. 

\begin{figure}[htb]
\begin{algorithm}[H]
\caption{Transition of the DPDA, Part 1: Loading the spins into memory \label{alg:DPDA_part1}}

\begin{algorithmic}[1]
\State Assert that $C \in \Sigma_{q \missingspinsymbol}$
\While{$C \neq \sun$}
	\Comment{Append the spin to the state}
	\State Assert that $C \in \Sigma_{q \missingspinsymbol}$
	\If{$|L_s|\leq 2k+1$} \Comment{First fill $L_s$}
		\State Append $L_s \gets L_sC$
	\Else \Comment{Keep the last $2k +1 $ spins in $L_s$}
		\State Remove the  leftmost spin in $L_s$
		\State Append $L_s \gets L_sC$
	\EndIf
	\While{$C \neq \rimo$}\Comment{Read the local interaction into $i$}
		\State Move to the next symbol
		\State Assert $C \in [p] \cup \{\lemo,\rimo,\addressdelimiter ,\square, \diamond\}$
		\State Set $i \gets iC$
	\EndWhile
	\State Assert $i \in \mathcal{I}_{1D}$
	\If{$|L_I|\leq k$} \Comment{First fill $L_I$}
		\State Append $L_I \gets L_Ii$
	\Else 
		\State Calculate the energy from the  leftmost local interaction and push it onto the stack \label{alg:DPDA_energy1}
		\State Remove the leftmost  local interaction in $L_I$
		\State Append $L_I \gets L_Ii$
	\EndIf
	\State Reset $i \gets \epsilon$
	\State Move to the next symbol
\EndWhile
\State Calculate the energy from the remaining local interactions and push them onto the stack \label{alg:DPDA_energy2}
\algstore{DPDA}
\end{algorithmic}\end{algorithm}
\end{figure}

\cref{alg:DPDA_part2}  validates the energy, similarly as \cref{alg:DFA_part2}. The automaton is just popping one symbol from the stack for each identical symbol of the input. If the stack is empty once the head arrives at $\newmoon$, the automaton accepts.

\begin{figure}[htb]
\begin{algorithm}[H]
	\caption{Transition of the DPDA, Part 2: Validating the energy}
	\label{alg:DPDA_part2}
	\begin{algorithmic}[1]
	\algrestore{DPDA}
	\State Move to the next symbol
	\While{$C \neq \newmoon$}
		\If{$C$ is equal to the top of the stack}
			\State Pop one element from the stack
		\Else
			\State Reject
		\EndIf
		\State Move to the next symbol
	\EndWhile
	\If{$C = \newmoon$ and $Z$ is on top of the stack}
		\State Accept
	\Else
		\State Reject
	\EndIf
\end{algorithmic}
\end{algorithm}
\end{figure}

It is easy to verify that this procedure can be implemented by a DPDA because 
(i) the head has finitely many states  (namely those given by \eqref{eq:headDPDA}), 
(ii) the head has deterministic rules (given by \cref{alg:DPDA_part1} and \cref{alg:DPDA_part2}), 
(iii) the head pushes / pops a bounded number of symbols on the stack at every step, and 
(iv) the head moves to the right at every step. 

The idea of this construction is illustrated in \cref{fig:dpda}. 

\begin{figure*}[th]\centering
	\includegraphics[width=.7\columnwidth]{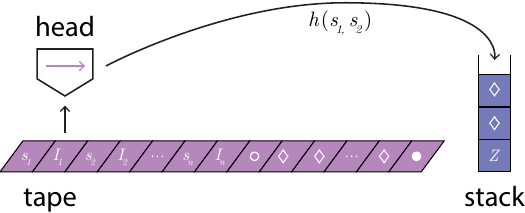}
	\caption{\textbf{A DPDA recognising the language of a 1D spin Hamiltonian.} 
	The tape initially contains the spin configuration intertwined with the local interactions, expressed as a string 
	$s_1, I_1, \ldots, s_n, I_n$, followed by a symbol $\sun$, an energy written in unary $\diamond \ldots \diamond$ and another symbol $\newmoon$. 
	This input string is in $L_{H}$  only if the energy is $H(s_1,I_1,\ldots, s_n,I_n)$. 
	A DPDA recognises this language as follows: it reads the state of the first two spins and local interactions (here for $k=2$), 
	calculates  $h(s_1, s_2)$, and stores this energy in the stack in unary (i.e.\ stores a number of diamonds  $\diamond$). 
	It proceeds similarly until the last pair of spins. Finally, it compares the energy stored in the stack with the value written on the tape, and accepts if and only if they coincide. 
	}
	\label{fig:dpda}
\end{figure*}

\paragraph{The language of a 1D spin Hamiltonian is not regular} We prove the statement with the following lemma: 

\begin{lemma}[Pumping lemma for regular languages, contrapositive form \cite{Ko97}] \label{lem:regularpumping}
	Let $L$ be a set of strings such that 
	for any $p\geq 1$ there exists strings $u$, $v$, $w$ with $uvw\in L$ and $|v|\geq p$, 
	such that for all decompositions of $v=xyz$ with $|y|\neq \epsilon$,
	there exists a $j \geq 0$ such that $uxy^jzw \notin L$. 
	Then $L$ is not regular.
\end{lemma}
For a given $p$ we choose strings 
$u=e(s) \sun$, $v=u(H(s))$ with $s \in \mathcal{D}$ such that $|H(s)|>p$, and $w=\newmoon$.  
 Note that  such an $s$ exists because  the image of $H$ is unbounded (otherwise it would be an effectively 0D spin Hamiltonian).
This ensures that the string $uvw=e(s) \sun u(H(s)) \newmoon$
is an element of $L_{H}$ with $|v|\geq p$. Any choice of $xyz=v$ and $j\neq 1$ will result in a malformed string $uxy^jz w\notin L_{H}$, since the spin configuration and the energy no longer match. Therefore $L_{H}$ is not regular.

\subsection{Proof of \cref{thm:main}  \ref{thm:dD} \label{sssec:proofofmain(iii)}}

\paragraph{The language of $d$D spin Hamiltonians is context-sensitive}

We now present a linear bounded automaton (LBA) which recognises the language of a $d$D spin Hamiltonian. For simplicity, we use an LBA with multiple tracks, where the input is written on track $1$, and the other tracks are used for calculations and bookkeeping tasks. Formally, the tape alphabet consists of $N_T$-tuples $(S_1, \ldots, S_{N_T})$ of symbols, where $N_T$ is the number of tracks.

To show that the following algorithm can be implemented by an LBA we only need to make sure that 
(i) the head has a finite number of states and tracks, and 
(ii) the head never moves past the end or before the beginning of the input string on the tape.

In the proof we will leverage the fact that concatenating LBAs gives rise to another LBA (\cref{lem:conc_lba}),
and that a set of basic operations, namely \textbf{Move},   \textbf{Copy}, \textbf{Add} and \textbf{Multiply}, can be carried out by an LBA. Let us present these  ingredients now.

\begin{lemma}[Concatenation of LBAs]\label{lem:conc_lba}
	Let $A_1$ and $A_2$ be two LBAs with the same alphabet $\Sigma$, and with $N_{T,1}, N_{T,2}$ tracks, respectively. 
	Let each automaton have two designated tracks called input and output. 
	Then there is an LBA $A_3$ which first runs $A_1$, uses the output of $A_1$ as the input of $A_2$, and then runs $A_2$.  Moreover 	$A_3$ uses at most $\max\{N_{T,1}, N_{T,2}\}$ tracks.
\end{lemma}

\begin{proof}
	We construct an LBA $A_3$ with $N_T = \max\{N_{T,1}, N_{T,2}\}$ tracks which first executes all the moves of $A_1$ and writes the output to some track $O_1$. When it reaches any of the final states $F_1$ of $A_1$, it transitions to a state \textbf{clean}, in which it traverses the tape once and resets all tracks  (except $O_1$) to the empty symbol. Then it transitions to the initial state
	of $A_2$, moves to its starting position and does all moves of $A_2$ with track $O_1$ as the input.
\end{proof}

Note that the same holds also for multiple input and output tracks, as a finite set of tracks $\{T_1,  \ldots , T_n\}$ can be redefined as a single track $T_{1,2,\ldots,m}$, where the symbols are   $n$-tuples of symbols from the other tracks.
Furthermore, if $A_2$ has no output (e.g.\ if it just checks some properties of the input) we can construct an LBA $A_2'$ which applies $A_2$ but outputs the input unchanged if the check succeeds. $A_2'$ first copies the input to an additional track (the new output) and then proceeds to apply $A_2$ to the input.

We now define the following basic operations, which can be done by an LBA as the following descriptions show: 

\textbf{Move $m$} takes an input track $I$ and moves all symbols to the right $m$ positions  (or to the left if $m<0$) for a fixed, predefined $m$. This is only possible when there are enough empty symbols on the right (left) end of the track. The machine then starts at the rightmost (leftmost) non-empty symbol, saves it, moves $m$ positions to the right and writes the symbol to the string. Then it moves $m+1$ positions to the left (right) and repeats the cycle. After the leftmost (rightmost) symbol is moved, the $m$ leftmost (rightmost) symbols are cleared.

\textbf{Copy} takes one input track $I$ and outputs $O$ with the same content as $I$ and leaves $I$ untouched. It first cleares $O$. Then it traverses the string and copies each symbol from $I$ to $O$.

\textbf{Add} takes two input tracks $I_1$ and $I_2$ containing valid unary encodings $u(a_1), u(a_2)$ (left aligned) and outputs $u(a_1 + a_2)$ on track $I_2$ while leaving $I_1$ untouched. Note that this only works if the tape has size at least  $|a_1+a_2|$.
The automaton works in the following way:
It first sets a bookmark on an auxiliary track at the first (leftmost) symbol of $I_1$ and saves the first symbol of $I_1$ to its state.
Then it moves to the last (non-empty) symbol of $I_2$ and appends or deletes a symbol, depending on the signs of $a_1$ and $a_2$. 
Next, it returns to the bookmark, moves the latter one position to the right, saves the corresponding symbol on track $I_1$ and moves again to the last symbol of $I_2$. This process is repeated until the bookmark reaches the end of $I_1.$

\textbf{Multiply} takes two input tracks $I_1$ and $I_2$ containing valid unary encodings $u(a_1), u(a_2)$ and outputs $u(a_1 \cdot a_2)$ on a track $O$ while leaving $I_1$ and $I_2$ untouched. Again, this works only if the tape has at least size $|a_1 \cdot a_2|$.
In the first step, the automaton copies $I_1$ to an auxiliary track $A_1$ and initialises $O$ to $\epsilon = u(0)$. Next, it uses \textbf{Add} to add $I_2$ to $O$ and then removes the rightmost symbol of $A_1$. This is repeated until $A_2$ is empty. Depending on the signs of $a_1$ and $a_2$ the sign of $O$ will be flipped in the end.

\paragraph{LBA recognising the language of a $d$D spin Hamiltonian} 
 
We denote the number of spins by $n$ and the length of the input by $N$. 
Note that the length of the encoding of any local interaction (in the input) is bounded by $|\gamma(I)|<N$ and that the number of possible local interactions (for a given $n$) $|\mathcal{I}_{n}|$ is bounded by a constant independent of $n$ (due to the bounded interaction range).

The LBA uses the following subroutines, which can be combined using \cref{lem:conc_lba}: 
\begin{enumerate} 
	\item \textbf{Validate pattern} \label{pag:LBA}
	\item \textbf{Count system size}
	\item \textbf{Compute side lengths}
	\item \textbf{Compute all possible local interactions} 
	\item \textbf{Validate the local interactions}
	\item \textbf{Evaluate local Hamiltonians}
	\item \textbf{Validate energy}
\end{enumerate}

\textbf{Validate pattern} validates the following pattern:
\begin{equation}\label{eq:LBA_pattern}
	\left( \Sigma_{q\missingspinsymbol}  [p] \lemo (U_\missingspinsymbol \addressdelimiter )^{k-1} U_\missingspinsymbol \rimo \right) ^+ \sun U \newmoon
\end{equation}
where $\Sigma^+=\Sigma^*\setminus\{\epsilon\}$ matches one or more elements from $\Sigma$,
 and $a+b$ means matching either $a$ or $b$, and we have defined the pattern matching an energy as 
 $U =\diamond^* + \square^* $ and an address as 
 $U_\missingspinsymbol =\diamond^* + \square^* + \missingspinsymbol $. 
Since this is just a regular expression,  there is an LBA which achieves this task. 
For the rest of the proof   we assume that the strings is well-formed. 

\textbf{Count system size}  counts the number of spins and writes the result on track $2$ in unary. First it adds a $\diamond$ on track $T_n$ for each spin symbol, or $\missingspinsymbol$ left of a symbol from $[p]$ in the input. Then it moves all those diamonds to the left to form a valid unary encoding. This only needs one additional track.

To \textbf{Compute the side lengths}
 recall that $n_1 = m \ell_1$ and so on until $n_d = m \ell_d$ for some known lattice $R_0 = (\ell_1, \ldots,\ell_d)$, and that $n = n_1 \cdots n_d =  m^d \ell$.
The machine first writes down $\ell_1$ on auxiliary track $A_{(0),1}$ and so on until $\ell_d$ on track $A_{(0),d}$. It also initialises  $T_{n,1}, \ldots T_{n,d}$ and $T_{N,1}, \ldots, T_{N,d}, A_c$ to $0$. Then it repeats the following steps:
\begin{enumerate}
	\item Increment $A_c$ by $1$
	\item For $i=1,\ldots, d :\ \mathbf{Add}(A_{(0),i}, T_{n,i})$
	\item \textbf{Copy} $T_{n,1}$ to $T_{N,1}$
	\item For $i=2,3,\ldots, d :\ \mathbf{Multiply}(T_{N,i-1}, T_{n,i}) \to T_{N,i}$ 
	\item $T_{N,d}$ now contains $(A_c)^d \cdot \ell $. Compare this to $n$ on track $T_n$.
	\item Exit the loop if $T_{N,d} = n$, continue if $T_{N,d} < n$, reject if $T_{N,d} > n$.
	\item Clear tracks $T_{N,i}$
\end{enumerate}
Tracks $T_{n,1},  \ldots, T_{n,d}$ now contain the side lengths $n_1,   \ldots, n_d$ and tracks $T_{N,i}$ contain $N_i = n_1 n_2 \cdots n_i$. Clear all other tracks. Note that the automaton can only run out of space if $T_{N,d} > n$, in which case it rejects. Furthermore, the number of tracks is independent of $n$.

\textbf{Compute all possible local interactions} writes out all possible local interactions explicitly. It needs at most $|\mathcal{I}_{n,j}| 2^k ( k d +1)$ tracks, named $T_{I, j}$ and  $T_{I, j, V_l^m}$, plus a fixed number of auxiliary tracks. 
On tracks $T_{I,j}$ the local interactions are written down, including all $2^k$ combinations of replacing an address with $\missingspinsymbol$. 
On track $T_{I,j, V_l^m}$ the corresponding moves are placed,  namely move $v_m$ of address $l$ of local interaction $j$ is put onto track $T_{I,j, V_l^m}$. Note that there are only finitely many possible local interactions, so they can essentially be hardwired into the machine.

Suppose that the automaton is to write local interaction number $j$ to track $T_{I,j}$. First, it puts the choice for the local Hamiltonian $\alpha \in [p]$ and the symbol $\lemo$ on track $T_{I,j}$.
Next, remember that the  addresses $A_l$ can be decomposed as 
$A_l = \sum_{i=0}^{d-1} v_i N_i$, 
where $N_0=1$ by definition. Note that $N_1, N_2, \ldots$ are given on tracks $T_{N,1}, T_{N,2}, \ldots$. 

For each address $l$ the machine first writes $v_m$ on tracks $T_{I,j, V_m^l}$ for $m=0,1,\ldots,d-1$, and then uses \textbf{Multiply} and \textbf{Add} to calculate $A_l$ on an auxiliary track. Then the address is copied to track $T_{I,j}$ followed by the symbol $\addressdelimiter$. If there is not enough space the machine continues with the next local interaction, as this local interaction is certainly not part of the input (otherwise there would be enough space). 
If some addresses are replaced with $\missingspinsymbol$ the machine still writes the moves on tracks $T_{I,j, V_m^l}$ but then skips the calculation and directly writes $\missingspinsymbol$ to the corresponding track.
Finally the machine clears all auxiliary tracks and continues with the next address.
After the last address the machine uses $\rimo$ instead of $\addressdelimiter$. Then it again resets all auxiliary tracks and continues with the next local interaction.

 \textbf{Validate the local interactions}  marks the $\lemo$ symbol of the first local interaction on an auxiliary track and then executes the following steps for each local interaction: 
\begin{enumerate}
	\item Copy the marked local interaction to an auxiliary track and move it all the way to the left
	\item Count the current index
	\item Decompose $j= \sum_{i=0}^{d-1} a_i N_i +1$
	\item Mark all valid candidates for the local interaction
	\item Validate that no edges are skipped (for each valid candidate)
	\item Clear all auxiliary tracks and move the bookmark to the next local interaction
\end{enumerate}

The first step looks for the bookmark and then copies the local interaction as described.
Next, the automaton counts $j$ (the position of the current local interaction) by adding $\diamond$ to an auxiliary track $T_{(j)}$ for each $\lemo$ symbol until (including) the bookmark, and then moving all $\diamond$ symbols to the left. In the third step, in order to  decompose $j$,  the machine first subtracts $1$ by deleting the rightmost $\diamond$, and  it initialises tracks $T_{a_0}, T_{a_1}, \ldots, T_{a_{d-1}}$ to hold $a_0, a_1, \ldots, a_{d-1}$. Then it subtracts $N_{d-1}$ from $T_{(j)}$ as long as $T_{(j)}\leq N_{d-1}$, and adds $\diamond$ to $T_{a_{d-1}}$ for each successful subtraction. It continues with $N_{d-2}, N_{d-3}$ until $N_0=1$. Once track $T_{(j)}$ is empty, tracks $T_{a_i}$ contain the decomposition of $j$.

The fourth step  compares symbol by symbol the copied local interaction to all the precomputed ones. As there are only finitely many precomputed local interactions the machine can use the states to keep track of the valid candidates. Note that while some precomputed local interactions can be equal because the addresses may have been replaced with $\missingspinsymbol$,  the corresponding vectors $\bm{v}$ are different.
Once the machine has traversed a local interaction in the input, 
the state contains a set of local interactions (the candidates) which match the given interaction. It is however not yet clear whether those interactions are allowed at the given position (i.e. whether $\missingspinsymbol$ is  in the correct places).

The penultimate step checks whether the missing spin symbols $\missingspinsymbol$ are placed correctly. For each candidate, the automaton validates the following condition for each address (numbered $l=1,2,\ldots,k$): 
The address is $\missingspinsymbol$ if and only if 
\begin{equation}\label{eq:proofdD:condI}
	\exists m \in \{0,1, \ldots , d-1\} : \ 0\leq a_m+v^l_m <n_{m+1}
\end{equation}
where $v^l_m$ is given on track $T_{I, j, V^l_m}$ and $n_m$ on track $T_{n,m}$. The automaton first checks whether the address specified in the local interaction is $\missingspinsymbol$,   adds the content of two tracks,  and compares the value to a third one $d$ times. As $n_m\leq n$ we have enough space for the addition. Namely,  $a_m+v^l_m \leq 2n \leq N $ for all valid cases. If the check succeeds, the machine continues with the next step. Otherwise the auxiliary tracks used for this step are cleaned and the machine continues with the next candidate. If no candidate is left the machine rejects.

Finally, the automaton clears the auxiliary tracks used and moves the bookmark to the next local interaction. If $\sun$ is reached, the automaton continues to the next step.

 \textbf{Evaluate the local Hamiltonian}  sets a bookmark at the first local interaction, and the spins specified by the addresses in the local interaction are copied to an auxiliary track $A_s$. To achieve this, the automaton copies an address to an auxiliary track and then moves a second bookmark as specified by the address (by moving the bookmark one spin and then deleting a symbol from the copied address). Then it copies the corresponding spin to $A_s$, clears the auxiliary track, resets the second bookmark and continues with the next address.
Next, the machine calculates $h_{\alpha}$ of the spins on track $A_s$. The corresponding energy is hardwired into the machine (as there are only finitely many local Hamiltonians with a finite domain each). The energy is written to a track $T_{E}$. If there are large intermediate energies multiple tracks are needed, but never more than the largest absolute local energy $\max_{\alpha \in [p]} \max_{x \in \Sigma_{q\missingspinsymbol}^k} |h_{\alpha}(x)|$ which is independent of $n$. Once one local Hamiltonian is evaluated, the bookmark is moved to the next interaction; if all local Hamiltonians are evaluated, the energy fits in one track (otherwise it is longer than the one given in the input and we reject).

Finally, \textbf{Validate energy} compares the energy on $T_{E}$ to that on the input. If they are equal it accepts, otherwise it rejects.

\paragraph{The language of a $d$D spin Hamiltonian is not context-free}
We prove this statement by exploiting the pumping lemma: 

\begin{lemma}[Pumping lemma for context-free languages, contrapositive form \cite{Ko97}]
	Let $L$ be a set of strings such that for any $p\geq 0$ there exists a string $z \in L$, such that for all decompositions of $z = uvwxy$ with $vx \neq \epsilon$ and $|vwx|\leq p$, there exists an $i \geq 0$ such that $uv^iwx^iy \notin L$.
	Then $L$ is not context-free.
	\label{lem:contextfreepumping}
\end{lemma}

Fix $p>0$ and choose any $s \in \mathcal{D}_{dD}$ such that $|s|\geq p$, and 
let $z = e(s)\sun u(H(s)) \newmoon$. 
The string $z$ can be split into substrings $ z = uvwxy$ in the following ways:
\begin{enumerate}
	\item $v$ or $x$ contain $\sun$ or $\newmoon$
	\item Either $\sun \in y$, or $\sun \in w$ and $x = \epsilon$
	\item $\sun \in w$ and neither $x$ nor $v$ are empty
	\item Either $\sun \in u$, or $\sun \in w$ and $v = \epsilon$
\end{enumerate}
For any of these cases there is $i$ such that $uv^iwx^iy \notin L_H$.
The first case is trivial as both $\newmoon$ and $\sun$ can only appear once, so any $i \neq 1$ yields a malformed string.
Case 4 is also trivial, as any $i \neq 1$ yields an energy not equal to $u(H(x))$.
For cases 2 and 3, assume that $uv^iwx^iy$ still matches the pattern given in \cref{eq:LBA_pattern} (otherwise the string is malformed) and that no single address is pumped (which would lead to an unbounded interaction range). 
Then the new string $uv^iwx^iy$ contains a different number of spins $n'(i)  = n + m  (i-1)$ for some $m \in \N$ which needs to be equal to the size of some lattice $n_0 l^d$ for some $l \in \N$ for the string to be in $L_H$. 
 However, the distance between consecutive numbers $n_0 l^d, \ n_0 (l+1)^d$ is at least 
 $ n_0 d l^{d-1}$. For $l$ large enough this is larger than $m$,  and hence there is at least one $i$ such that $n_0 l^d < n'(i) < n_0 (l+1)^d$, which implies that $s \notin L_H$.

\paragraph{Proof of \cref{thm:main} \ref{thm:dD} for non-rectangular lattices}

Any $d$D spin Hamiltonian $H'$ on a non-rectangular lattice can be seen as a $d$D spin Hamiltonian $H$ on a rectangular lattice with a non-trivial unit cell. 
Explicitly, the spin alphabet $\Sigma_q$ in $H$ now contains tuples of the spin alphabet of $H$', where the size of the tuple is given by the number of spins in the unit cell. Since $\Sigma_q$ is in the domain of the local Hamiltonians, these are redefined accordingly. This gives rise to an $H$ defined on a rectangular lattice. Therefore the above proof showing that its language $L_H$ is context-sensitive and not context-free applies here too.

\subsection{Proof of \cref{thm:main} \ref{thm:all} \label{sssec:proofofmain(iv)}}

\paragraph{The language of an all-to-all spin Hamiltonian is context-sensitive}
We provide an LBA to recognise the language of all-to-all spin Hamiltonians by reusing parts of the LBA of the $d$D case. 
This LBA performs the following subroutines: 
\begin{enumerate}
	\item \textbf{Validate pattern}
	\item \textbf{Validate local interactions}
	\item \textbf{Evaluate local Hamiltonians}
	\item \textbf{Validate energy}
\end{enumerate}
In the first step, the pattern is adjusted to allow for a list of local interactions. 
The second step  works  as follows: 
For each spin $j\leq n-k$ the automaton executes \cref{alg:LBA_alltoall}.  
In \cref{alg:LBA_alltoall:next_addresses},  the next address is defined by the order of  \eqref{eq:order}. 
Following this order can be done by an LBA as follows: 
first the automaton attempts to increment $A^{k-1}$ by 1, which is only possible if $A^{k-1} < n-j$. 
If this does not succeed, it attempts to increment  $A^{k-2}$ by 1, 
which is only possible if  $A^{k-2} < n-j-1$.  
If incrementing $A^{k-2}$ by 1 succeeds, it sets $A^{k-1}=A^{k-2}+1$. 
If it does not succeed, it continues with  $A^{k-3}$ recursively. 
In the third step, evaluating the local Hamiltonians uses the same subroutine as in the $d$D case, 
as this does not depend on the number of local interactions per spin. 
The comment about the maximal number of tracks needed also holds in this case. 
That is, at most $\max_{\alpha \in [p]} \max_{x \in \Sigma_{q\missingspinsymbol}^k} |h_{\alpha}(x)|$ tracks are needed.
The fourth step is the same as for the $d$D case. 

\begin{figure}[htb]
\begin{algorithm}[H]
	\caption{Validating local interactions for all-to-all spin Hamiltonians}
	\label{alg:LBA_alltoall}
	\begin{algorithmic}[1]
	\For{all spins}
		\State Put the index of the current spin $j$ on track $T_s$
		\State Initialise working tracks $A^0,  \ldots, A^{k-1}$ to $0,\ldots, k-1$.
		\For{all local interactions of spin $j$}
			\State Copy the addresses of the local interaction to tracks $B^0, \ldots, B^{k-1}$.
			\State Assert $B^l=A^l$ for all $l=0, \ldots,k-1$.
			\If{$(A^0, A^1, \ldots, A^{k-1}) = (0,n-j-k+2,\ldots,n-j)$ }
				\State Go to the next spin
			\Else
				\State Set $A^0,  \ldots, A^{k-1}$ to the next address \label{alg:LBA_alltoall:next_addresses}
			\EndIf
		\EndFor
	\EndFor
\end{algorithmic}
\end{algorithm}
\end{figure}

\paragraph{The language of an all-to-all spin Hamiltonian is not context-free}
To prove this statement we again use the pumping lemma (\cref{lem:contextfreepumping}). 
	So fix $p>0$, choose any $s \in \mathcal{D}_{\mathrm{all}}$ such that $n=|s|> p$ and $n>k$, and let $z = e(s)\sun u(H(s)) \newmoon$. 
	The string $z$ can be split into substrings $ z = uvwxy$ in the following ways:
	\begin{enumerate}
		\item $v$ or $x$ contain $\sun$ or $\newmoon$
		\item Either $\sun \in y$, or  $\sun \in w$ and $x = \epsilon$ 
		\item $\sun \in w$ and neither $x$ nor $v$ are empty
		\item Either $\sun \in u$, or $\sun \in w$ and $v = \epsilon$
	\end{enumerate}
	As above, cases 1 and 4 are trivial. 
	For cases 2 and 3, denote by $\ell_{\mathrm{in}}(z)$ the length of the string $z$ until $\sun$ (i.e.\ without the energy). As the only choice for the input are the spins and the local Hamiltonians (which are reflected by one symbol each), $\ell_{\mathrm{in}}(z)$ only depends on $n$ (for $z \in L_H$), and so we write $\ell_{\mathrm{in}}(z)=\ell_{\mathrm{in}}(n)$ for any valid $z$ with $n$ spins. Comparing inputs of size $n$ and $n+1$ we see that no local interactions are removed. 
	That is, all local interactions present in an input of size $n$ are also present in an input of size $n+1$ (up to different choices of the local Hamiltonians). Hence $\ell_{\mathrm{in}}(n+1)-\ell_{\mathrm{in}}(n)$ can be lower bounded by the difference in the number of local interactions (as each local interaction needs more than one symbol to be encoded).
	This difference is given by  
	\begin{equation}
		{{n+1}\choose{k}} - {{n}\choose{k}} = {{n}\choose{k-1}} \geq  
		 n > p
	\end{equation}
	where we have used that $k\geq 2$ (cf.\ \cref{def:all-to-all}), 
	$n>k$ and $n>p$. 
	Choosing $i=2$ implies that $|uv^2wx^2y| - |uvwxy|\leq p$ and thus also $\ell_{\mathrm{in}}(uv^2wx^2y) -\ell_{\mathrm{in}}(uvwxy) \leq p$. Hence $uv^2wx^2y$ is not in $L_H$.

\subsection{Proof of \cref{thm:main} with other boundary conditions \label{sssec:proofofmainother}}

We now show that all statements of \cref{thm:main} also hold for other boundary conditions. 
First, a $0D$ spin Hamiltonian  with periodic boundary conditions also gives rise to a finite language, which is thus regular. 
On the other hand, we claim that an effectively $0D$ spin Hamiltonians $H$ with periodic boundary conditions  must be the constant function. To prove this claim, note that the configurations can  be viewed as rings of spins and interactions. Such a ring can be doubled by cutting it open at any position, inserting a copy and putting it back together. Formally, take any configuration $x=(s_1, I_1, s_2, I_2, \ldots, s_n, I_n)$ and map it to $x^2 \coloneqq (s_1, I_1, \ldots, s_n, I_n, s_1, I_1, \ldots, I_n)$. Then, $H(x^2)= 2 H(x)$ as each interaction $I_j$ occurs exactly twice in $x^2$ and the values of the spins addressed do not change. Hence, for the image of $H$ to be bounded it is necessary that $H\equiv 0$.
It trivially follows that there is a DFA accepting the language, as it only needs to validate whether the string is well-formed and whether the energy is zero.

The complexity of 1D spin Hamiltonians with periodic boundary conditions does not change either. To see this, we slightly modify the DPDA so that it saves the $2k+1$ first spins and $k$ first local interaction in its head, 
and only evaluates them once it reaches $\sun$, using also the saved $k$ last spins. The pumping lemma can be applied in exactly the same way, as only the energy was pumped.

Finally, for $dD$ spin Hamiltonians with other boundary conditions we can apply the pumping lemma in the same way, as it exploits the fact that the string needs to represent a lattice of points. 
To show that the language is context-sensitive the LBA is modified as follows. 
For any address and any periodic edge $i$ the automaton checks whether it would jump over this edge (by checking $0\leq a_i+v_i\leq n_{i+1}$). If a jump occurs, the automaton calculates the new address. For all directions with a jump this is done by $v'_i= a_i+ v_i \mod{n_{i+1}}$; for the directions with open boundary conditions, a rotation can be applied (e.g. by $v'_l= n_{l+1} - 2 a_l$), and the rest remains unchanged. The new address is then given by $\sum_i v'_i N_i$. Note that this address is never longer than $n$ as all $0\leq v'_i < n_{i+1}$, and all operations can be performed by an LBA  (addition, subtraction, modulo). 
The other parts of the algorithm of the LBA stay unchanged. This modification thus results again in an LBA and  the language remains context-sensitive.

\section{Proof of \cref{pro:computcompl} (Computational complexity of $L_H$) \label{ssec:proofofcomputcompl}}

The proof of \cref{thm:main} has provided automata to recognise $L_H$ for each  $H$ in each of the considered classes. 
Since each of these automata are a particular case of a deterministic Turing machine, 
showing that each of these algorithms runs in polynomial time will prove the result. 
Note that this will not result in optimal running times. 
Throughout we use the $O$ notation, where we say that an algorithm is $O(f(N))$ if its running time  grows not faster in the input length $N$ than some function $f(N)$.

\subsection{Proof of \cref{pro:computcompl} \ref{pro:computcompl:0D} and \ref{pro:computcompl:eff0D}}

For $L_H$ of a 0D and an effectively 0D spin Hamiltonian, the deterministic finite automaton's (DFA) running time is linear in the input size, 
because---apart from a constant overhead---each transition moves the head one symbol to the right and the computation stops once all symbols are read. 

\subsection{Proof of \cref{pro:computcompl} \ref{pro:computcompl:1D}}

For $L_H$ of a 1D spin Hamiltonian,  the computational complexity of the deterministic pushdown automaton (DPDA) is linear in the input size. This is because a DPDA can be seen as a DFA with a stack, and the transitions which push/pop multiple symbols are unproblematic, as the number of pushes/pops is bounded independently of the system size. Thus, this DPDA runs in linear time. 
When simulating the DPDA with a Turing machine, the Turing machine does not have a stack and needs to simulate it by using a parallel track or writing and working at the end of the input. 
For each spin in the input, the Turing machine needs to run to the end of the input (or to the parallel track), which takes $O(N)$ time, where $N$ is the size of the input. Since there are $n$ spins, and $N=O(n)$ (since the input consists of the spins and the local interactions, as well as the energy, which needs at most $|E_{\max}-E_{\min}|n$ symbols, where $E_{\max}$ and  $E_{\min}$ are the maximal and minimal energies of all the local Hamiltonians), the Turing machine runs in $O(N^2)$ time.

\subsection{Proof of \cref{pro:computcompl} \ref{pro:computcompl:dD}}

In the $d$D case, the computational complexity of the linear bounded automaton (LBA) is polynomial in the input size $N$.  
First note that both the number of spins and the energy are smaller than the input size, i.e. $n < N$ and $|u(H(x))| < N$.
We show that each of the subroutines' running time is at most polynomial in $N$. In between the subroutines there might be some cleanup tasks (e.g.\ cleaning a track, moving the head to the beginning), which are all linear in $N$.  
Recall the subroutines of the LBA (page \pageref{pag:LBA}):
\begin{enumerate}
	\item \textbf{Validate pattern}
	\item \textbf{Count system size}
	\item \textbf{Compute side lengths}
	\item \textbf{Compute all possible local interactions} 
	\item \textbf{Validate the local interactions}
	\item \textbf{Evaluate local Hamiltonians}
	\item \textbf{Validate energy}
\end{enumerate}
Steps 1., 2., and 7. only need one run along the tape (or part of it), so they are linear in $N$.

Both \textbf{Add} and \textbf{Multiply} are polynomial in their arguments, as $\textbf{Add}(a,b)$ uses $O(a\cdot b)$ steps and 
$\textbf{Multiply}(a,b)$ uses $O(a^2b^2)$. 
Note also that $a^d$ for a fixed constant $d$ uses  $O(a^{2d})$ steps. 

\textbf{Compute side lengths} computes $m^d \ell$ for $m=1, 2, \ldots m_{\max} = (n/\ell)^{1/d}$. This can be bounded by $m_{\max}$ calculations of $m_{\max}^d \ell$. As we have shown above, this is polynomial in $m_{\max}$. As $m_{\max} =(n/\ell)^{1/d}$ it also is polynomial in $n$ and hence in $N$. 

\textbf{Compute all possible local interactions} consists of multiplications and additions of $n_1, \ldots, n_{d-1}$ and some constants for each local interaction. As the number of local interactions $|\mathcal{I}_{n,j}|$ is constant, the running time of the subroutine is polynomial in $N$.

\textbf{Validate the local interactions} copies the current local interaction to another track and compares it to all precomputed local interactions. Copying can be done in $O(N^2)$ and comparing uses at most $N$ steps. Then, the current index is decomposed and for each valid candidate,  where additional additions and comparisons are done. Decomposing the index is achieved by a hardcoded set of divisions which are polynomial in $n$, and the number of candidates is bounded by $|\mathcal{I}_{n,j}|$. 

\textbf{Evaluate local Hamiltonians} uses  $O(k N^2)$ steps to copy the relevant spins to a separate track and a linear amount of steps to read the spins and write down the energy. Copying the energy to the correct track again uses $O(k N^2)$ steps. \textbf{Validate the local interactions} and \textbf{Evaluate local Hamiltonians} will each be called $n$ times. This introduces some polynomial overhead for the counter, rendering everything still polynomial.
In conclusion, all the subroutines are polynomial in $N$, hence the full algorithm is polynomial in $N$.

In order to account for other boundary conditions, some addresses need to be replaced during the \textbf{Evaluate local Hamiltonians} subroutine. Namely, an address  decomposed as $\sum_i v_i N_i$ is replaced by $\sum_i v'_i N_i$,  where  $v'_i = a_i + v_i \mod (n_{i+1})$ where $a_i$ is the decomposition of the spin index $j = \sum_i a_i N_i +1$. All those operations can be done in polynomial time.

\subsection{Proof of \cref{pro:computcompl} \ref{pro:computcompl:all}}
To obtain a bound for the all-to-all case we only need to consider the modified \textbf{Validate the local interactions} subroutine, as the rest is identical to the $d$D case.
This modified subroutine calculates the addresses for each local interaction. There are $n \choose k$ local interactions, and for each local interaction the automaton compares the calculated interaction to that given on the input, and then increments the calculated address to the next one. 
Both steps are linear in $n$, and hence the complexity of the subroutine is   $O(n^{k+1})$, i.e.\ polynomial in $N$.

\section{Proof of \cref{pro:binary} (Classification in the Chomsky hierarchy with binary encoding)} 
\label{ssec:proofofbinary}

 \cref{pro:binary} can be broken down into the following statements: 
	\begin{enumerate}[before=\leavevmode, label=(\roman*),ref=(\roman*)]
	\item \label{pro:binary:0D} 
	If $H$ is a 0D (\cref{def:0D}) or an effectively 0D spin Hamiltonian (\cref{def:eff0D}) then $L_H^{\mathrm{b}}$ (\cref{def:language}) is regular. 
	\item \label{pro:binary:1D} 
	If $H$ is a 1D spin Hamiltonian (\cref{def:1D}) then $L_H^{\mathrm{b}}$ is context-sensitive and not context-free. 
	\item \label{pro:binary:dD} 
	If $H$ is a $d$D spin Hamiltonian with $d\geq2$ (\cref{def:dD})   then $L_H^{\mathrm{b}}$ is context-sensitive and not context-free. 
	\item \label{pro:binary:all} 
	If $H$ is an all-to-all spin Hamiltonian (\cref{def:all-to-all}) then $L_H^{\mathrm{b}}$ is context-sensitive and not context-free. 
	\end{enumerate}
	In addition, \ref{pro:binary:0D}, \ref{pro:binary:1D} and  \ref{pro:binary:dD} hold for open, periodic or quasi-periodic boundary conditions along any dimension. 
	
	This proof is structured following these items: we will prove \ref{pro:binary:0D}, \ref{pro:binary:1D},  \ref{pro:binary:dD}  and \ref{pro:binary:all}, and finally the result for other boundary conditions. 
	
\subsection{Proof of \cref{pro:binary} \ref{pro:binary:0D}}

We follow the same line of reasoning as for \cref{thm:main} \ref{thm:0D}. The binary language of a 0D spin Hamiltonian is finite and hence regular. For the effectively 0D Hamiltonian we only need to slightly modify the DFA. When reaching $\sun$ the automaton transitions to a state $E = b(H(x))$ (of which there are finitely many because the energy is bounded), and  continues to compare the string $b(H(x))$ in its head with the encoded energy given in the input. It accepts if and only if they match.  

\subsection{Proof of \cref{pro:binary} \ref{pro:binary:1D}}

For 1D spin Hamiltonians we use Ogden's lemma to show that $L_H^{\mathrm{b}}$ is not context-free. Then we modify the LBA from the $d$D case so that it accepts $L_H^{\mathrm{b}}$, showing that the latter is context-sensitive.

\begin{lemma}[Ogden's lemma \cite{Og68}]\label{lem:ogden}
Let $L$ be a context-free language. Then there exists an integer $p\geq 1$ such that for any string $z\in L$ with $|z|\geq p$ the following holds: For any way of selecting at least $p$ symbols in $z$, there is a decomposition $z=uvwxy$ such that 
\begin{enumerate}
\item[(i)] $vx$ contains at least one of the selected symbols, 
\item[(ii)] $vwx$ contains at most $p$ selected symbols, and 
\item[(iii)] $uv^jwx^jy \in L$ for all $j\geq 0$.
\end{enumerate}
\end{lemma}

So assume $L_{H}^{\mathrm{b}}$ is context-free.    
Then for any $p\geq 1$, there exists an element $z\in L_{H}^{\mathrm{b}}$ such that the binary representation of the energy, $E_b=b(H(x))$, contains at least $p$ symbols (otherwise, the energy would be bounded and the Hamiltonian would   be effectively 0D). Our selection of $p$ symbols (\cref{lem:ogden}) are the $p$ most significant digits of $E_b$. 
We claim that for any decomposition $z=uvwxy$  for which condition (i) holds, condition (iii) does not hold. 
Assuming (i), $v x$ contains at least one symbol from $E_b$.
Now, the string $uv^j wx^jy$ with $j>1$ is  shifting the most significant digit of $E_b$ at least $j-1$ positions to the left, so that the corresponding energy grows by at least a factor of $2^{j-1}$. To obtain a valid string, $v$ needs to lie to the left of $\sun$, otherwise the new energy will not match the energy of the input configuration. For the string to still be well-formed (i.e. a valid encoding of an input) $v$ can either lie entirely within one address or it can contain a number $m$ of spins and local interactions (i.e.\ $m$ of each). In the former case, choosing $j>k$ ensures that $|v^j|>k$ and that the corresponding address is longer than $k$ and hence invalid.
In the latter case, adding $(j-1)m$ spins and local interactions can change the energy by at most $  (|E_{\max}| + |E_{\min}| ) (2k + (j-1)m)$, where $E_{\min}$ and $E_{\max}$ are the minimal (maximal) energy of any local Hamiltonian,  
the $2k$ term  accounts for the local interactions not in $v$ which  contain spins from within $v$,  
and the $(j-1)m$ term accounts for all new local interactions. 
Thus, the energy in the second part of the input grows faster than the energy encoded in the first part. It follows that $L_{H}^{\mathrm{b}}$ is not context-free.

To show that $L_{H}^{\mathrm{b}}$ of a 1D spin Hamiltonian is context-sensitive we only need to modify the \textbf{Validate energy} subroutine of the LBA of the proof of \cref{thm:main} \ref{thm:dD}.
All other subroutines also work in the 1D case;  for example, to list all possible local interactions, the LBA can just write them out directly, as they do not depend on $n$.
The modified \textbf{Validate energy} calculates the binary encoding $b(T_E)$ of the calculated energy $T_E$, before comparing the former to the input.
To calculate the binary encoding, the machine calculates $i$, $2^i$ and $2^{i+1}$ (in unary, on three separate tracks) for $i=1,2,\ldots$ until $2^{i+1}$ is larger than $T_E$ (or reaches the end of the tape). Then it adds a $1$ on position $i$ on a fourth track $T_b$ and subtracts $2^i$ from $T_E$. Those steps are repeated until $T_E$ is empty. Then the machine fills the gaps between the $1$s (and the start of the tape) on track $T_b$ with $0$s, adds the sign symbol at the end ($+$ or $-$, depending on the sign of $T_E$) and inverts the resulting string. Track $T_b$ now contains $b(T_E)$ which is compared with $b(E)$ given in the input.

\subsection{Proof of \cref{pro:binary} \ref{pro:binary:dD} and \ref{pro:binary:all}}

The statement of $d$D  and all-to-all spin Hamiltonians is easily proven by using the modified \textbf{Validate energy} subroutine given in the proof of \ref{pro:binary:1D}. The resulting LBA will  recognise $L_H^{\mathrm{b}}$ of a $d$D or all-to-all spin Hamiltonian, implying that $ L_H^{\mathrm{b}}$ is context-sensitive. 
To show that $L_H^{\mathrm{b}}$ is not context free we  use the same proof as for the unary case, as it does not depend on the encoding of the energy.

\subsection{Proof of \cref{pro:binary} with other boundary conditions}

The boundary conditions are independent of the choice of a binary encoding of the energy, in the sense that  the former only cares about which spins interact whereas the latter only comes into play once all the local interactions are evaluated. So the modifications to the automata needed for the other boundary conditions and those needed for the binary encoding of the energy can be applied together. Additionally, the proofs involving variations of the pumping lemma are independent of the boundary conditions---a fact that we also used in the proof of \cref{thm:main} with other boundary conditions.

\section{Proof of \cref{pro:langtime} (Classification in the Chomsky hierarchy of $L_U$) \label{ssec:proofoflangtime}}

\subsection{Proof of \cref{pro:langtime} \ref{pro:langtime:0D}}

	If $H$ is a 0D spin Hamiltonian, its domain is finite, so  $L_U$ is a finite language, and thus trivially regular. 

\subsection{Proof of \cref{pro:langtime} \ref{pro:langtime:special}}

	Let $H$ be an effectively 0D spin Hamiltonian, or a 1D spin Hamiltonian, in either case with one spin symbol and fixed interactions. We first claim that $L_U$ is not regular, and use  the pumping lemma for regular languages (\cref{lem:regularpumping}) to prove it. 
So let $p>0$ and choose a string $z=e(s)\sun e(s) \newmoon$ with $s \in \mathcal{D}$ and $ |e(s)|>p$. Then choose any decomposition $z=uvw$ with $v \neq \epsilon$, $i\neq 1$ and define the pumped string $z_p = u v^i w$.
	As both $\sun$ and $\newmoon$ can only appear once, they cannot be part of $v$. Hence the pumping can only happen either on the left or on the right of $\sun$. In both cases, the parts before and after $\sun$ (not including $\newmoon$) will have different lengths. Thus, $z_p$ is not an element of $L_U$.

To show that $L_U$ is deterministic context-free, we provide a deterministic pushdown automaton (DPDA) recognising the language.
	The automaton needs to check that:
	\begin{enumerate}
		\item The string is well-formed, i.e.\ it is of the form $e_b^n \sun e_b^m \newmoon$ where $e_b$ is the encoding of the (single valid) spin symbol followed by the encoding of the local interaction, and where $m, n \in \N$.
		\item $n = m$.
	\end{enumerate}
	The first part can be done by a deterministic finite-state automaton (DFA) running `on top' of the DPDA, as $e_b$ is finite and arbitrary repetitions of finite words form a regular language. The addition of $\sun$ and $\newmoon$ does not change the complexity.
	The second task can be done by a DPDA. The automaton puts one symbol onto the stack for each symbol in the input until it reaches $\sun$. Then it pops one symbol from the stack for each symbol in the input until it reaches $\newmoon$. The automaton accepts only if the last symbol of the stack is popped when reaching $\newmoon$.

\subsection{Proof of \cref{pro:langtime} \ref{pro:langtime:other}}

	Let $H$ be a $d$D spin Hamiltonian with $d\geq 1$ which does not satisfy the conditions of case  \ref{pro:langtime:special}, 
	or an effectively 0D spin Hamiltonian  not satisfying the conditions of case  \ref{pro:langtime:special}, 
	or an all-to-all spin Hamiltonian. 
	
	Let us first assume that $H$ have $q>1$ or unfixed interactions. 
	We use the pumping lemma for context-free languages (\cref{lem:contextfreepumping}) to show that $L_U$ is not context-free.
	So let $p>0$, and choose a configuration $s \in \mathcal{D}$ with an even number of spins $|s| > p$, 
	such that the first $|s|/2$ spins have a different value than the last $|s|/2$ (if there is only one possible spin symbol we can use different local interactions), and choose the string $z= e(s)\sun e(s)\newmoon$.
	Choose any decomposition $z=uvwxy$ with $|vwx|\leq p$, $|vx|>0$ and any $i\neq 1$ and define the pumped string $z_p = u v^i w x^i y$.
	We can assume that the pumped string is well-formed (i.e.\ of the form $e(\sigma) \sun e(\sigma') \newmoon$ for $\sigma,\sigma'$ in $\mathcal{D}$), as any malformed string is surely not in $L_U$.	
	Another necessary condition for $z_p$ to be in $L_U$ is that $|e(\sigma)| = |e(\sigma')|$,
	 which is possible only  if $\sun \in w$ and $|v|=|x|<p/2$. 
	Since $|e(s)| > |s|>p$, the spin values in $v$ are different from the spin values in $x$. 
	Say, for example, that the spins in the first half of $s$ have value $1$ and the spins in the second half have value $2$. 
	Then $v$ contains only spins with value $2$ and $x$ only such with value $1$. 
	In the pumped string, $\sigma$ contains more spins with value $2$ than with value $1$, 
	whereas $\sigma'$ contains more spins with value $1$ than $2$. 
	Hence $z_p$ is not in $L_U$.
	
	Now assume that  $H$ be an all-to-all spin Hamiltonian with $q=1$ and fixed interactions. 
	As any string $e(s)$ for $s \in \mathcal{D}$ is  such that the first half is different to the second half (as the number of local interactions per spin changes), we  choose $s \in \mathcal{D}$ with $|s| > p$ and follow the same argument as directly above to conclude that $L_U$ is not context-free. 

	Finally assume that $H$ be a $d$D spin Hamiltonian with $d\geq2$ with $q=1$ and fixed interactions. 
	In this case, we apply the pumping lemma for context-free languages as in the proof of \cref{thm:main} \ref{thm:dD}, since it only relies in the structure of $e(s)$ for $s\in \mathcal{D}$. It follows that $L_U$ is not context-free.

	To show that $L_U$ is context-sensitive, 
	we provide a linear bounded automaton (LBA) recognising the language, which  checks the following subroutines: 
	\begin{enumerate}
		\item The string is well-formed, i.e.\ it is of the form $e(x)\sun e(y) \newmoon$ where $e(x)$ and $e(y)$ are valid encodings (i.e.\ they lie in the image of $e$).
		\item $x$ is a valid configuration, i.e.\ $x \in \mathcal{D}$.
		\item $e(y)$ is equal to $e(x)$.
	\end{enumerate}
	Subroutine 1 can be carried out by a DFA by checking a pattern, as  in the proof of \cref{thm:main} \ref{thm:dD}.
	Subroutine 2 has in fact been implemented for all cases by the respective automata in the proof of \cref{thm:main}. Since all  those automata are LBAs or weaker types of machines, this subroutine can be run on an LBA.
	Subroutine 3 is easily run on an LBA: The automaton copies the part of the input after $\sun$ to a new track and moves it all the way to the left. Then it moves along the tape comparing the two tracks one symbol at the time (rejecting if they differ), until it reaches $\sun$ on the input. It accepts if $\newmoon$ is on the copied track and rejects otherwise.

\end{document}